\documentclass{IEEEtran}

\usepackage[table,dvipsnames]{xcolor}
\usepackage{url}
\usepackage{cite}
\usepackage{soul}
\usepackage{tikz}
\usepackage[normalem]{ulem}
\usepackage{array}
\usepackage{float}
\usepackage{amsthm}
\usepackage{amsmath}
\usepackage{amssymb}
\usepackage{nccmath}
\usepackage{optidef}
\usepackage{enumitem}
\usepackage{graphicx}
\usepackage{multirow}
\usepackage{pgfplots}
\usepackage{thmtools}
\usepackage{clipboard}
\usepackage{microtype}
\usepackage{tikzscale}
\usepackage{algorithm}
\usepackage{mathtools}
\usepackage{nicematrix}
\usepackage{subcaption}
\usepackage{algorithmic}
\usepackage{thm-restate}
\usepackage[T1]{fontenc}
\usepackage[font=small,labelfont=bf]{caption}

\allowdisplaybreaks

\newcommand{\bo}[1]{\boldsymbol{#1}}
\newcommand{\ur}[1]{\mathrm{#1}}
\renewcommand{\H}{\ur{H}}
\newcommand{\T}{\ur{T}}

\newcommand{\btr}[1]{\ur{tr}\negmed\left(#1\right)}
\newcommand{\bdet}[1]{\ur{det}\negmed\left(#1\right)}

\newcommand{\bmin}[1]{\ur{min}\negmed\left\{#1\right\}}

\newcommand{\bnorm}[1]{\left\lVert#1\right\rVert}
\newcommand{\rank}[1]{\ur{rank}\negmed\left(#1\right)}

\newcommand{\blogdet}[1]{\ur{log}_\ur{e}\ur{\,det}\negmed\left(#1\right)}
\newcommand{\bloggdet}[1]{\ur{log}_2\ur{\,det}\negmed\left(#1\right)}
\newcommand{\E}[1]{\ur{E}\negmed\left[#1\right]}

\newcommand{\bdiag}[1]{\ur{diag}\negmed\left(#1\right)}

\newcommand{\bO}[1]{\mathcal{O}\negmed\left(#1\right)}

\def\negstrip#1 #2\relax{-#1}
\newcommand*{\negmed}{\mkern-\thinmuskip}

\theoremstyle{general} \newtheorem{theorem}{Theorem}
\theoremstyle{general} \newtheorem{lemma}[theorem]{Lemma}
\theoremstyle{general} 
\theoremstyle{general} \newtheorem{proposition}{Proposition}
\theoremstyle{general} \newtheorem{p-corollary}{Corollary}[proposition]
\theoremstyle{general} 
\theoremstyle{general} 
\theoremstyle{remark}  \newtheorem{remark}{Remark}

\begin{document}
	\title{Downlink MIMO-RSMA with Successive Null-Space Precoding}
	\author{\IEEEauthorblockN{Aravindh Krishnamoorthy\rlap{\textsuperscript{\IEEEauthorrefmark{2}\IEEEauthorrefmark{1}}}\,\,\,  and Robert Schober\rlap{\textsuperscript{\IEEEauthorrefmark{2}}}\\
	\IEEEauthorblockA{\small \IEEEauthorrefmark{2}Friedrich-Alexander-Universit\"{a}t Erlangen-N\"{u}rnberg, Germany\\
	\IEEEauthorrefmark{1}Fraunhofer Institute for Integrated Circuits (IIS) Erlangen, Germany}}\thanks{This work was supported by BMBF Project 6G-RIC (Project ID: 16KISK023). This article was presented in part at the Workshop on Rate-Splitting (Multiple Access) for 6G held during the IEEE Intl. Commun. Conf. (ICC) 2021 \cite{Krishnamoorthy2021a}.}}
	\maketitle
	
	\begin{abstract}
		In this paper, we consider the precoder design for an underloaded or critically loaded downlink multi-user multiple-input multiple-output (MIMO) communication system. We propose novel precoding and decoding schemes which enhance system performance based on rate splitting at the transmitter and single-stage successive interference cancellation at the receivers. The proposed successive null-space (SNS) precoding utilizes linear combinations of the null-space basis vectors of the successively augmented MIMO channel matrices of the users as precoding vectors to adjust the inter-user-interference experienced by the receivers. We formulate a non-convex weighted sum rate optimization problem for the precoding vectors and the associated power allocation for the proposed SNS-based MIMO-rate-splitting multiple access (RSMA) scheme. We obtain a suboptimal solution for this problem via successive convex approximation. Moreover, we study the robustness of the proposed precoding scheme to imperfect channel state information (CSI) at the base station via derivative-based sensitivity analysis. Our analysis and simulation results reveal the enhanced performance and robustness of the proposed SNS-based MIMO-RSMA scheme over several baseline multi-user MIMO schemes, especially for imperfect CSI.
	\end{abstract}
	
	\section{Introduction}
	Exploiting multiple-input multiple-output (MIMO) communication channels is crucial for the 5th generation (5G) and beyond communication systems in order to meet the ever-increasing demand for mobile data services \cite{Saad2020}. To this end, motivated by the recent improvements regarding the signal processing capabilities of user terminals, MIMO concepts capitalizing on successive interference cancellation (SIC) at the receiver such as MIMO non-orthogonal multiple access (MIMO-NOMA) \cite{Ding2017,Krishnamoorthy2020,Makki2020,Krishnamoorthy2021,Liu2022} and MIMO rate-splitting multiple-access (MIMO-RSMA) \cite{Mao2018,Zhou2020,Dizdar2020a,Dizdar2021} have been proposed in the context of beyond 5G downlink communication systems.
		
	The MIMO-NOMA literature has mostly focused on overloaded communication systems where the combined number of antennas at the user terminals is larger than that at the BS, a scenario in which MIMO-NOMA is most beneficial, see references in \cite{Ding2017,Makki2020}. However, in order to exploit the potential performance gains, for $K$ users, MIMO-NOMA requires in total $K(K-1)/2$ stages of SIC at the receivers, which incurs a high complexity if there are more than a few users. Consequently, the authors in \cite{Krishnamoorthy2020, Krishnamoorthy2021} proposed low-complexity two-user MIMO-NOMA schemes which exploited the null spaces of the MIMO channels of the users. However, for underloaded or critically loaded communication systems, where the combined number of receive antennas is smaller than or equal to the number of transmit antennas, these schemes have a considerable performance gap to the dirty paper coding (DPC) upper bound (UB) \cite{Vishwanath2003}.
	
	On the other hand, although several MIMO-RSMA schemes for overloaded systems have been proposed, e.g., see \cite{Piovano2016,Joudeh2017,Kaulich2021}, the MIMO-RSMA literature has mostly focused on underloaded or critically loaded communication systems \cite{Mao2018,Zhou2020,Dizdar2020a,Dizdar2021}. MIMO-RSMA can be seen as a generalization of MIMO-NOMA and space division multiplexing, which supports multiple hierarchical coded layers, see \cite{Mao2018} for details. Nevertheless, in this paper, in order to limit the decoding complexity, we restrict ourselves to single-layer MIMO-RSMA \cite{Mao2018}, which necessitates only a single stage of SIC at the receivers. Single-layer MIMO-RSMA utilizes a single common message (CM), intended for all users, along with the conventional MU-MIMO private streams for the users, which are precoded using a multi-user linear precoder (MU-LP). This hybrid structure, in addition to its enhanced performance, allows for a trade-off between performance and robustness to imperfect channel state information (CSI) at the base station (BS) via suitable power allocation (PA) between the CM and the private streams, see \cite[Proof of Theorem 1]{Joudeh2016},\cite{Mao2020} for details. Nevertheless, the robustness of MIMO-RSMA with respect to imperfect CSI can be further improved by careful design of the MU-LP.  Moreover, careful PA is crucial for realizing the performance benefits of MIMO-RSMA.
		
	Conventional MU-LP schemes such as zero forcing (ZF) \cite{Wiesel2008}, regularized zero forcing (RZF) \cite{Peel2005}, and block diagonalization (BD) \cite{Spencer2004} have been utilized for MIMO-RSMA, see, e.g., \cite{Flores2019, Flores2020}. However, in practice, these schemes cause a large performance gap to the DPC UB, see, e.g., \cite[Fig. 8]{Krishnamoorthy2021}. Furthermore, as ZF and BD precoders are designed to avoid inter-user-interference (IUI), imperfect CSI causes a significant performance degradation. On the other hand, the low-complexity linear successive allocation (LISA) scheme proposed in \cite{Guthy2009, Utschick2018} eliminates IUI through a combination of successive projections and ZF, see \cite{Guthy2009} for more details. LISA has been studied for overloaded MIMO-RSMA systems with minimum rate constraints in \cite{Kaulich2021}. The successive encoding and successive allocation method (SESAM) in \cite{Tejera2006} eliminates IUI by using a combination of successive projections and coding, e.g., DPC. However, eliminating IUI via coding incurs a high computational complexity. Alternatively, the MU-LP matrices may also be directly computed via numerical optimization. However, as we shall show via analysis, this necessitates a high computational complexity and, in most cases, the resulting precoders lack robustness to imperfect CSI. Hence, the development of novel MU-LP designs achieving high performance and robustness with reasonable computational complexity is crucial for realizing the full potential of MIMO-RSMA.
	
	To this end, in this paper, we propose a novel successive null space (SNS) precoder for the private messages in MIMO-RSMA. The proposed SNS precoder utilizes linear combinations of the null-space basis vectors of the successively augmented MIMO channel matrices of the users as precoding vectors to adjust the IUI experienced by the receivers for enhancing the performance and robustness. Furthermore, we propose an optimization based framework for obtaining the linear precoding vectors and the PA. Moreover, we study the impact of CSI imperfection using derivative-based sensitivity analysis (SA), see \cite{Kucherenko2016} and references therein, which yields important insights into the impact of imperfect CSI on performance. Interestingly, this analysis is also applicable when the CSI error statistics are unknown.

	This paper builds upon the conference version in \cite{Krishnamoorthy2021a}, which proposed SNS-based MIMO-RSMA for perfect CSI. In this paper, we extend the scheme to imperfect CSI, present a performance analysis based on SA, and provide derivations for the matrix-valued first-order approximations used in \cite{Krishnamoorthy2021a}. Furthermore, we significantly expand the simulation section to verify the performance of SNS-based MIMO-RSMA. The main contributions of this paper can be summarized as follows.
	\begin{itemize}
		\item We present the proposed SNS-based precoders for downlink MIMO-RSMA.
		\item We formulate an optimization problem for maximization of the weighted sum rate (WSR) of SNS-based MIMO-RSMA, and solve it via successive convex approximation (SCA) \cite{Razaviyayn2014} to obtain a feasible lower bound (LB) on the performance.
		\item We analyze the impact of imperfect CSI on the performance of the proposed SNS-based MIMO-RSMA scheme via derivative-based SA and numerical simulations.
	\end{itemize}
	
	Moreover, we compare the performance LB for the proposed scheme with several baseline schemes including the DPC UB \cite{Vishwanath2003}, RZF and ZF precoding \cite{Peel2005,Wiesel2008}, BD precoding \cite{Spencer2004}, and BD-based MIMO-RSMA \cite{Flores2019} via computer simulations for both perfect and imperfect CSI. Our results show that, for both perfect and imperfect CSI, the proposed SNS-based MIMO-RSMA scheme outperforms RZF and ZF precoding, BD precoding, and BD-based MIMO-RSMA. We note that the proposed SNS-based MIMO-RSMA necessitates SIC at the users, which increases the receiver complexity compared to conventional MU-LP schemes. However, SIC at the users is deemed feasible in future generation receivers with improved signal processing capabilities, see, e.g., \cite{Ding2017,Liu2022}.
	
	The remainder of this paper is organized as follows. In Section \ref{sec:prel}, we present the system model, and briefly describe the baseline precoding schemes. The proposed SNS precoding and decoding schemes are provided in Section \ref{sec:pd}. WSR optimization for perfect CSI is studied in Section \ref{sec:mwsr}. The derivative-based SA for the SNS precoders and the adaptation of the WSR optimization problem for imperfect CSI are described in Section \ref{sec:mwsri}. Simulation results are presented in Section \ref{sec:sim}, and the paper is concluded in Section \ref{sec:con}.
		
	\emph{Notation:}
	Boldface capital letters $\bo{X}$ and boldface lower case letters $\bo{x}$ denote matrices and vectors, respectively. $\bo{X}^\T$, $\bo{X}^\H$, $\btr{\bo{X}}$, and $\bdet{\bo{X}}$ denote the transpose, Hermitian transpose, trace, and determinant of matrix $\bo{X}$, respectively. $\bnorm{\bo{x}}$ and $\bnorm{\bo{X}}$ denote the Euclidean norm and the induced Euclidean norm of vector $\bo{x}$ and matrix $\bo{X},$ respectively. Furthermore, $\bo{X}^+$ and $\bo{X}^\frac{1}{2}$ denote the Moore-Penrose pseudo-inverse and a square root of matrix $\bo{X},$ respectively. $\mathbb{C}^{m\times n}$ and $\mathbb{R}^{m\times n}$ denote the sets of all $m\times n$ matrices with complex-valued and real-valued entries, respectively.  $\bo{I}_N$ denotes the $N\times N$ identity matrix, and $\bo{0}$ denotes the all zero matrix of appropriate dimension. The circularly symmetric complex Gaussian (CSCG) distribution with mean vector $\bo{\mu}$ and covariance matrix $\bo{\Sigma}$ is denoted by $\mathcal{CN}(\bo{\mu},\bo{\Sigma})$; $\sim$ stands for ``distributed as''. $\E{\cdot}$ denotes statistical expectation.
	
	\section{Preliminaries}
	\label{sec:prel}
	In this section, we present the downlink MIMO system model, the imperfect CSI model, and brief reviews of RZF, ZF, and BD precoding as well as BD-based MIMO-RSMA.
	
	\subsection{System Model}
	\label{sec:sm}
	We consider an \emph{underloaded} or \emph{critically loaded} downlink MU-MIMO communication system comprising a BS with $N$ transmit antennas and $K$ users equipped with $M_k,k=1,\dots,K,$ antennas, such that $N \geq \sum_{k=1}^{K} M_k.$

	The MIMO symbol vectors of the users are constructed using RS, as described in the following. First, the private message of the $k$-th user, $k=1,\dots,K,$ is encoded into a MIMO symbol vector $\bo{s}_{k} \in \mathbb{C}^{M_k\times 1},$ where $\E{\bo{s}_{k} \bo{s}_{k}^\H} = \bo{I}_{M_k}, \forall\,k.$ Additionally, a CM including messages to all downlink users is encoded into a MIMO symbol vector $\bo{s}_{\ur{c}} \in \mathbb{C}^{M\times 1}, \E{\bo{s}_{\ur{c}} \bo{s}_{\ur{c}}^\H} = \bo{I}_M,$ where $M = \bmin{M_k,k=1,\dots,K}.$ Symbol vectors $\bo{s}_{\ur{c}}$ and $\bo{s}_{k},\forall\,k,$ are assumed to be statistically independent. Furthermore, parameter $\eta_{\ur{c},k},$ $0 \leq \eta_{\ur{c},k} \leq 1, k=1,\dots,K,$ $\sum_{k=1}^{K}\eta_{\ur{c},k} = 1,$ is used to assign a fraction $\eta_{\ur{c},k}$ of the available bits in the CM to user $k.$ $\eta_{\ur{c},k} = 0$ implies that user $k$'s message is excluded from the CM.
	
	Next, the common MIMO symbol vector, $\bo{s}_{\ur{c}},$ and the private MIMO symbol vector, $\bo{s}_{k}, k=1,\dots,K,$ are precoded using linear precoders $\bo{P}_{\ur{c}} \in \mathbb{C}^{N\times M}$ and $\bo{P}_{k} \in \mathbb{C}^{N\times M_k},$ respectively, and the precoded symbol vectors are superimposed to obtain the transmit signal $\bo{x} = \bo{P}_{\ur{c}}\bo{s}_{\ur{c}} + \sum_{k = 1}^{K} \bo{P}_{k}\bo{s}_{k}.$ The transmit power constraint at the BS is as follows:
	\begin{align}
		\btr{\bo{P}_{\ur{c}} \bo{P}_{\ur{c}}^\H} + \sum_{k = 1}^{K} \btr{\bo{P}_{k} \bo{P}_{k}^\H} \leq P_\T, \label{eqn:p1}
	\end{align}
	where $P_\ur{T}$ denotes the available transmit power. The transmitter is schematically illustrated in Figure \ref{fig:sm}. Let $\frac{1}{\sqrt{\mathstrut L_k}} \bo{H}_k \in \mathbb{C}^{M_k\times N}$ denote the MIMO channel matrix between the BS and user $k,$ where scalar $L_k$ models the path loss between the BS and user $k,$ and $\bo{H}_k$ models the small scale fading. Here, we assume that all MIMO channel matrices have full row rank\footnote{We note that a row-rank deficient MIMO channel matrix can be transformed into a full row-rank matrix with fewer effective receive antennas via singular value decomposition, see, e.g., \cite[App. C]{Scutari2009}.}. Then, the received signal at user $k$ is given by
	\begin{align}
		\bo{y}_k &{}= \frac{1}{\sqrt{\mathstrut L_k}} \bo{H}_k\bo{x} + \bo{z}_k \nonumber\\ &{}= \frac{1}{\sqrt{\mathstrut L_k}} \bo{H}_k \Big(\bo{P}_{\ur{c}}\bo{s}_{\ur{c}} + \sum_{k' = 1}^{K} \bo{P}_{k'}\bo{s}_{k'}\Big) + \bo{z}_k, \label{eqn:sm:y_k}
	\end{align}
	where $\bo{z}_k \in \mathbb{C}^{M_k\times 1} \sim \mathcal{CN}(\bo{0}, \sigma^2\bo{I}_{M_k})$ denotes the complex additive white Gaussian noise (AWGN) vector at user $k.$
	
	While we assume perfect CSI in Sections \ref{sec:pd} and \ref{sec:mwsr}, imperfect CSI is considered in Section \ref{sec:mwsri}. In the following, we describe the adopted imperfect CSI model.
	
	\subsection{Imperfect CSI Model}
	\label{sec:impmimo}
	In Section \ref{sec:mwsri}, we assume that only quantized and outdated MIMO channel matrices $\bar{\bo{H}}_k,k=1,\dots,K,$ are available at the BS. Nevertheless, we assume that the BS knows\footnote{This assumption is motivated by the slow variation of the path loss $L_k$ and the resulting low feedback requirement.} scalar $L_k$ perfectly, and the users know their own MIMO channel matrices perfectly\footnote{In practice, users can estimate their own MIMO channel matrices frequently and accurately, e.g., by exploiting orthogonal, high-power pilot sequences transmitted by the BS. Furthermore, residual estimation errors can be incorporated into the receiver noise model \cite{Wang2007,Eraslan2013}, thereby allowing the use of the received signal model in (\ref{eqn:sm:y_k}) also for imperfect CSI.}. We model the estimated MIMO channel matrices at the BS as follows\footnote{We note that several models for capturing the impact of estimation errors have been proposed in the literature, see, e.g., references in \cite{Joudeh2016}. Nevertheless, in this paper, we restrict ourselves to the versatile additive error model, which is widely used in the RSMA literature, e.g., \cite{Mao2018,Dizdar2021,Joudeh2016,Flores2019, Flores2020}.}:
	\begin{align}
		\frac{1}{\sqrt{\mathstrut L_k}} \bar{\bo{H}}_k = \frac{1}{\sqrt{\mathstrut L_k}} \bo{H}_k + \frac{1}{\sqrt{\mathstrut L_k}} \Delta\bo{H}_k, \quad k=1,\dots,K, \label{eqn:hkd}
	\end{align}
	where $L_k$ and $\bo{H}_k$ denote the actual path loss coefficient and the MIMO channel matrix of user $k,$ respectively, as earlier, and $\Delta\bo{H}_k \in \mathbb{C}^{M_k\times N}$ models the estimation error. Furthermore, we assume $\|\Delta\bo{H}_k\| < \|\bar{\bo{H}}_k\|,$ i.e., the CSI error is small compared to the CSI itself.
		
	\begin{figure*}
		\begin{minipage}{0.45\textwidth}
			\centering
			\includegraphics[width=\textwidth]{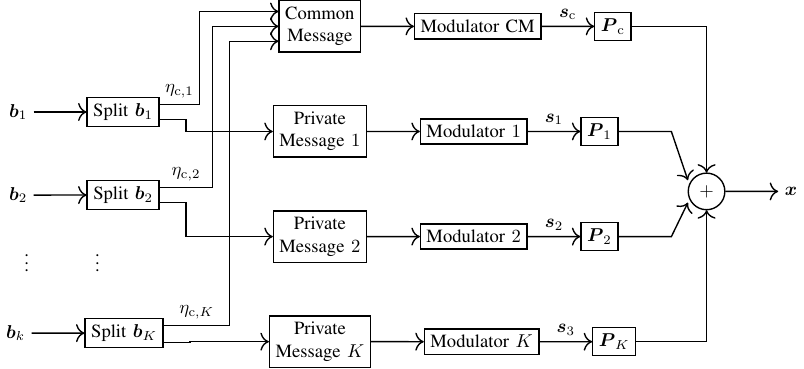}
			\caption{Schematic diagram of a single-layer MIMO-RSMA transmitter illustrating the generation of the transmit signal $\bo{x}$ from the user bit-streams $\bo{b}_k,k=1,\dots,K.$}
			\label{fig:sm}
		\end{minipage}%
		\hspace{0.05\textwidth}%
		\begin{minipage}{0.45\textwidth}
			\centering
			\resizebox{0.85\textwidth}{!}{%
				\renewcommand{\arraystretch}{1.5}
				\begin{tabular}{cc|c|c|c|ccc|c|cl}
					\hline
					Symbol/User      && $1$  & $2$ & $3$ &&        && $K$ && Precoder \\\hline\hline
					$\bo{s}_{\ur{c}}$ && \cellcolor{Turquoise}D    & \cellcolor{Turquoise}D   &  \cellcolor{Turquoise}D  &\cellcolor{Turquoise}&\cellcolor{Turquoise}$\cdots$&\cellcolor{Turquoise}& \cellcolor{Turquoise}D   && $\bo{P}_{\ur{c}}$ \\\hline
					$\bo{s}_{1}$ && \cellcolor{LimeGreen}S    & \cellcolor{Goldenrod}I   & \cellcolor{Goldenrod}I  &\cellcolor{Goldenrod}&\cellcolor{Goldenrod}$\cdots$&\cellcolor{Goldenrod}& \cellcolor{Goldenrod}I   && $\bo{X}_{1}^\frac{1}{2}$ \\\hline
					$\bo{s}_{2}$ &&      & \cellcolor{LimeGreen}S   &  \cellcolor{Goldenrod}I  &\cellcolor{Goldenrod}&\cellcolor{Goldenrod}$\cdots$&\cellcolor{Goldenrod}& \cellcolor{Goldenrod}I   && $\bo{\Psi}_{2} \bo{X}_{2}^\frac{1}{2}$ \\\hline
					$\vdots$  &&      &     &     &&        &&     && $\vdots$ \\\hline
					$\bo{s}_{K}$ &&      &     &     &&&& \cellcolor{LimeGreen}S   && $\bo{\Psi}_{K} \bo{X}_{K}^\frac{1}{2}$ \\\hline
			\end{tabular}}
			\medskip
			
			\resizebox{0.5\textwidth}{!}{%
				\begin{tabular}{|cl|}
					\hline
					\cellcolor{LimeGreen}S & Privately decoded symbol \\
					\cellcolor{Turquoise}D & Commonly decoded symbol \\
					\cellcolor{Goldenrod}I & Symbols treated as noise\\
					\hline
			\end{tabular}}
			\caption{Schematic diagram of the proposed SNS precoding and decoding schemes.}
			\label{fig:decoding}
		\end{minipage}
	\end{figure*}
	
	\subsection{Regularized Zero Forcing based Precoding and Decoding}
	\label{sec:zf}
	RZF precoding is a robust technique widely used for underloaded or critically loaded downlink MIMO communication systems \cite{Peel2005}. For conventional RZF precoding, no CM is used, i.e., $\bo{P}_\ur{c} = \bo{0},$ and the precoders for the private messages $\bo{P}_k,k=1,\dots,K,$ are chosen as:
	\begin{align}
		\begin{bmatrix}
			\bo{P}_1 & \dots & \bo{P}_K
		\end{bmatrix} &{}= \bo{H}^\H\big(\bo{H}\bo{H}^\H + \alpha\bo{I}_{\sum_{k=1}^{K}M_k}\big)^{-1} \nonumber\\
					&\quad\times\begin{bmatrix}
					(\bo{D}_1^{\textrm{ZF}})^\frac{1}{2} & & \\ & \ddots & \\ & & (\bo{D}_K^{\textrm{ZF}})^\frac{1}{2}
					\end{bmatrix}, \label{eqn:pkzf}
	\end{align}
	where $\bo{H} = \begin{bmatrix} \bo{H}_1^\T & \dots & \bo{H}_K^\T \end{bmatrix}^\T$ and $\bo{D}_k^{\textrm{ZF}} \in \mathbb{R}^{M_k \times M_k} = \bdiag{p_{k,1}^{\textrm{ZF}},\dots,p_{k,M_k}^{\textrm{ZF}}}$ is the diagonal power-allocation matrix chosen such that $\btr{\begin{bmatrix}
				\bo{P}_1 & \dots & \bo{P}_K
			\end{bmatrix} \begin{bmatrix}
		\bo{P}_1 & \dots & \bo{P}_K
	\end{bmatrix}^\H} \leq P_\T.$ Here, $\alpha$ is a design parameter which is used to regularize $\bo{H}\bo{H}^\H.$ In \cite{Peel2005}, the optimal value of $\alpha$ for maximizing the received signal-to-noise ratio (SNR) was found to be $\alpha = \frac{\sum_{k=1}^{K} M_k \sigma^2}{P_\T},$ which we use in our simulations in Section \ref{sec:sim}. When $\alpha=0,$ RZF reduces to ZF \cite{Wiesel2008}, in which case (\ref{eqn:sm:y_k}) can be simplified to $\bo{y}_k = \frac{1}{\sqrt{\mathstrut L_k}} (\bo{D}_k^{\textrm{ZF}})^\frac{1}{2}\bo{s}_{k} + \bo{z}_k.$ Hence, assuming perfect CSI at the BS, ZF precoding completely eliminates self-interference and IUI for the private symbols, greatly simplifying the receiver design.
	
	\subsection{Block Diagonalization Precoding and Decoding}
	\label{sec:bd}
	For conventional BD precoding \cite{Spencer2004}, similar to RZF precoding, no CM is used, i.e., $\bo{P}_\ur{c} = \bo{0},$ and the precoders for the private messages $\bo{P}_k,k=1,\dots,K,$ are chosen as follows. In accordance with \cite{Spencer2004}, let $\bo{\Psi}_k^{\textrm{BD}} \in \mathbb{C}^{N\times M_{k}}$ denote a matrix whose columns are the unit-length basis vectors of the null space of the following augmented matrix: $\bo{F}^{\textrm{BD}}_k = \begin{bmatrix}
			\bo{H}_{1}^\T &  \dots & \bo{H}_{k-1}^\T & \bo{H}_{k+1}^\T &  \dots & \bo{H}_{K}^\T 
		\end{bmatrix}^\T\negthickspace.$
	Furthermore, let $\bo{H}_k\bo{\Psi}_k^{\textrm{BD}} = \bo{U}_k^{\textrm{BD}} \bo{\Sigma}_k^{\textrm{BD}} \bo{V}_k^{\textrm{BD}}$ via singular value decomposition (SVD), where $\bo{U}_k^{\textrm{BD}}, \bo{V}_k^{\textrm{BD}} \in \mathbb{C}^{M_k\times M_k}$ are unitary matrices and $\bo{\Sigma}_k^{\textrm{BD}} \in \mathbb{R}^{M_k\times M_k}$ is a diagonal matrix containing the $M_k$ singular values of $\bo{H}_k\bo{\Psi}_k^{\textrm{BD}}$ on the main diagonal. Then, the precoder for the private message of the $k$-th user is chosen as $\bo{P}_k = \bo{\Psi}_k^{\textrm{BD}} \bo{V}_k^{\textrm{BD}} (\bo{D}_k^{\textrm{BD}})^\frac{1}{2},$ where $\bo{D}_k^{\textrm{BD}} \in \mathbb{R}^{M_k \times M_k} = \bdiag{p_{k,1}^{\textrm{BD}},\dots,p_{k,M_k}^{\textrm{BD}}}$ is the diagonal power-allocation matrix chosen such that $\sum_{k=1}^{K} \sum_{l=1}^{M_k} p_{k,l}^{\textrm{BD}} \leq P_\T.$ With BD precoding, (\ref{eqn:sm:y_k}) can be rewritten as follows:
	\begin{align}
		\bo{y}_k &= \bo{H}_k\bo{\Psi}_k^{\textrm{BD}} \bo{V}_k^{\textrm{BD}} (\bo{D}_k^{\textrm{BD}})^\frac{1}{2}\bo{s}_{k} + \bo{z}_k. \label{eqn:y_k_bd}
	\end{align}
	As can be seen from (\ref{eqn:y_k_bd}), assuming perfect CSI at the BS, BD precoding eliminates IUI for the private symbols completely, thereby simplifying the receiver design. Self-interference can also be eliminated without performance loss by using $(\bo{U}_k^{\textrm{BD}})^\H$ as the detection matrix at user $k,$ see \cite{Spencer2004} for details.
	
	\subsection{BD-based MIMO-RSMA}
	\label{sec:bdrsma}
	An extension of BD precoding for MIMO-RSMA was presented in \cite{Flores2019} where a scalar CM symbol $s_c$ was assumed. Nevertheless, the scheme in \cite{Flores2019} can be generalized to MIMO CM in a straightforward manner. For BD-based MIMO-RSMA using a generic common precoder $\bo{P}_{\ur{c}}$ for CM, (\ref{eqn:sm:y_k}) can be rewritten as:
	\begin{align}
		\bo{y}_k = \frac{1}{\sqrt{\mathstrut L_k}} \bo{H}_k \bo{P}_{\ur{c}}\bo{s}_{\ur{c}} + \bo{H}_k\bo{\Psi}_k^{\textrm{BD}} \bo{V}_k^{\textrm{BD}} (\bo{D}_k^{\textrm{BD}})^\frac{1}{2}\bo{s}_{k} + \bo{z}_k.
	\end{align}

	\subsection{Shortcomings of RZF, ZF, and BD Precoding}
	\label{sec:shortcoming}
	While RZF, ZF, and BD precoding are popular schemes, they suffer from several performance deficiencies. For Gaussian MIMO channels, the performance of ZF precoding deteriorates as $\sum_{k=1}^{K} M_k \to N$ due to the required inversion of a poorly conditioned matrix. The performance of RZF precoding, which mitigates this deterioration by regularizing the matrix inversion, is sensitive to the choice of design parameter $\alpha,$ and has poor performance at low SNRs \cite{Sung2009}. On the other hand, for both underloaded and critically loaded systems, BD precoding suffers from poor performance because the precoding vectors, $\bo{\Psi}_k^{\textrm{BD}},$ lie in a low ($M_k$)-dimensional subspace of the available $N$-dimensional vector space. This effect is especially pronounced in underloaded systems and for correlated MIMO channels. Hence, BD precoding is unable to fully exploit the degrees of freedom (DoFs) available at the BS. Thus, the performance of RZF, ZF, and BD precoding is substantially lower than that of DPC. Furthermore, BD-based MIMO-RSMA inherits the shortcomings of BD precoding described above, and its performance also has a significant gap to the DPC UB. Moreover, as ZF and BD precoding are designed to avoid IUI, they lack robustness to imperfect CSI knowledge at the BS. 

	Hence, in the following, we develop the proposed SNS MU-LP scheme that overcomes these limitations and provides enhanced performance and robustness.
	
	\section{Proposed Successive Null-Space based Precoding and Decoding}
	\label{sec:pd}
	In this section, we present the proposed SNS precoding and decoding schemes. The proposed schemes utilize a fixed precoder structure based on the null spaces of the MIMO channels of the users. In this section, we assume perfect CSI. The changes required to account for imperfect CSI at the BS are detailed in Section \ref{sec:mwsri}.
	
	\subsection{Proposed Precoding Scheme}
	\label{sec:precoding}
	Let $\bo{\Psi}_k \in \mathbb{C}^{N\times N_k},$ $N_k = N-\sum_{k' = 1}^{k-1} M_{k'},$ denote a matrix whose columns are the unit-length basis vectors of the null space of the following augmented matrix:
	\begin{align}
		\bo{F}_k = \begin{bmatrix}
		\bo{H}_{1}^\T & \bo{H}_{2}^\T & \dots & \bo{H}_{k-1}^\T
		\end{bmatrix}^\T\negthickspace, \label{eqn:augh}
	\end{align}
	with the convention $\bo{\Psi}_{1} = \bo{I}_N.$ The proposed \emph{SNS precoder} for user $k$ is constructed as follows:
	\begin{align}
		\bo{P}_{k} &= \bo{\Psi}_{k} \bo{X}_{k}^\frac{1}{2}, \label{eqn:pk}
	\end{align}
	where $\bo{X}_{k} \in \mathbb{C}^{N_k\times N_k}$ is a symmetric, positive semi-definite matrix with rank $M_k,$ which is to be optimized. Since matrices $\bo{X}_{k},k=1,\dots,K,$ have rank $M_k,$ they can be factorized as $\bo{X}_{k} = \bo{X}_{k}^\frac{1}{2} (\bo{X}_{k}^\frac{1}{2})^\H,\forall\,k,$ where the rectangular matrix factors $\bo{X}_{k}^\frac{1}{2} \in \mathbb{C}^{N_k\times M_k}$ have $M_k$ columns each. On the other hand, the precoder for CM, $\bo{P}_{\ur{c}},$ is an unconstrained full matrix with no predefined structure. Now, since
	\begin{align}
		\btr{\bo{P}_{k} \bo{P}_{k}^\H} &= \btr{\bo{\Psi}_{k} \bo{X}_{k} \bo{\Psi}_{k}^\H} = \text{tr}\Big({\underbrace{\bo{\Psi}_{k}^\H \bo{\Psi}_{k}}_{\bo{I}_{N_k}} \bo{X}_{k}}\Big) 	= \btr{\bo{X}_{k}}, \label{eqn:trsimp}
	\end{align}
	for $k=1,\dots,K,$ the condition in (\ref{eqn:p1}) can be rewritten for SNS-based MIMO-RSMA as follows:
	\begin{align}
		\btr{\bo{P}_{\ur{c}} \bo{P}_{\ur{c}}^\H} + \sum_{k = 1}^{K} \btr{\bo{X}_{k}} \leq P_\T \label{eqn:p2}.
	\end{align}
	
	\begin{remark}
		As $\bo{P}_{k}$ utilizes linear combinations of the basis vectors in $\bo{\Psi}_{k},$ symbol vectors precoded with $\bo{P}_{k}$ do not cause IUI to users $k'=1,\dots,k-1,$ however, they cause IUI to the remaining users $k'=k+1,\dots,K.$ Furthermore, for user $k,$ the BD basis vectors, $\bo{\Psi}_k^{\textrm{BD}},$ lie in a subspace spanned by the SNS basis vectors $\bo{\Psi}_{k}.$ Hence, in $\bo{P}_{k},$ the IUI can be adjusted by judiciously combining the basis vectors in $\bo{\Psi}_{k}$ using $\bo{X}_{k}^\frac{1}{2}.$ In the extreme case, IUI can be completely eliminated by selecting only the BD subspace, see our analysis in Section \ref{sec:bdcomp} for details.
	\end{remark}
	
	\begin{remark}
		As seen from (\ref{eqn:augh}) and (\ref{eqn:pk}), the precoder matrices depend on the user labels $1,\dots,K.$ Users with lower indices have more degrees of freedom to choose their precoder $\bo{P}_{k}.$ Consequently, the user rates depend on the user labeling. Hence, ideally, user rate optimization should be carried out over all permutations of user labels.
	\end{remark}
	
	Based on the SNS precoder matrices given above, (\ref{eqn:sm:y_k}) can be rewritten as follows:
	\begin{align}
		\bo{y}_k &= \frac{1}{\sqrt{\mathstrut L_k}}\bo{H}_k\bo{P}_{\ur{c}}\bo{s}_{\ur{c}} + \frac{1}{\sqrt{\mathstrut L_k}}\bo{H}_k \sum_{k'=1}^{k} \bo{\Psi}_{k'} \bo{X}_{k'}^\frac{1}{2}\bo{s}_{k'} + \bo{z}_k. \label{eqn:y_r}
	\end{align}
	
	\begin{remark}
		Note that the proposed SNS-based precoding and decoding can also be utilized without CM by setting $\bo{P}_\ur{c} = \bo{0}.$
	\end{remark}
	
	\subsection{Decoding Scheme}
	\label{sec:decoding}
	At user $k,$ decoding and SIC are performed as follows. First, using $\bo{y}_k,$ the common symbol vector $\bo{s}_{\ur{c}}$ is decoded treating the contributions of $\bo{s}_{k'},k'=1,\dots,k,$ as noise\footnote{The common symbol vector must also be decoded when $\eta_{c,k} = 0.$}. Following successful decoding\footnote{Decoding is assumed to be always successful as the symbols are transmitted at or below their achievable rates. In practice, powerful codes that can closely approach these achievable rates can be utilized.} of $\bo{s}_{\ur{c}},$ its contribution is eliminated, resulting in the signal:
	\begin{align}
		\bo{y}_k' &= \bo{y}_k -  \frac{1}{\sqrt{\mathstrut L_k}}\bo{H}_k\bo{P}_{\ur{c}}\bo{s}_{\ur{c}} 
		\nonumber\\&{}=  \frac{1}{\sqrt{\mathstrut L_k}} \bo{H}_k \bo{\Psi}_{k} \bo{X}_{k}^\frac{1}{2}\bo{s}_{k} + \frac{1}{\sqrt{\mathstrut L_k}} \bo{H}_k \sum_{k'=1}^{k-1} \bo{\Psi}_{k'} \bo{X}_{k'}^\frac{1}{2}\bo{s}_{k'} + \bo{z}_k. \label{eqn:y_r1}
	\end{align}
		
	Next, exploiting $\bo{y}_k',$ symbol vector $\bo{s}_{k}$ is decoded treating $\bo{s}_{k'},k'=1,\dots,k-1,$ as noise. The proposed SNS precoding and decoding strategies are schematically illustrated in Figure \ref{fig:decoding}.
	
	\subsection{Achievable Rate}
	Based on (\ref{eqn:y_r}) and (\ref{eqn:y_r1}), the achievable rate of user $k$ can be obtained as follows. The rate of private symbol vector $\bo{s}_{k}$ is given by
	\begin{align}
		R_{k} &= \log_2\det\Bigg(\bo{I}_{M_k} + \frac{1}{L_k}\bo{H}_k\bo{\Psi}_{k} \bo{X}_{k} \bo{\Psi}_{k}^\H\bo{H}_k^\H \nonumber\\&\quad\times\Big[\sigma^2\bo{I}_{M_k} + \frac{1}{L_k}\bo{H}_k \Big(\sum_{k'=1}^{k-1}\bo{\Psi}_{k'} \bo{X}_{k'} \bo{\Psi}_{k'}^\H\Big)\bo{H}_k^\H\Big]^{-1}\Bigg). \label{eqn:rk2}
	\end{align}
	Next, for the common symbol vector, $\bo{s}_\ur{c},$ the achievable rate of user $k$ is:
	\begin{align}
		R_{k,\ur{c}} &= \log_2\det\Bigg(\bo{I}_{M_k} + \frac{1}{L_k}\bo{H}_{k}\bo{Q}_{\ur{c}}\bo{H}_{k}^\H \nonumber\\&\quad\times\Big[\sigma^2\bo{I}_{M_{k}} + \frac{1}{L_k}\bo{H}_{k} \Big(\sum_{k'=1}^{k}\bo{\Psi}_{k'} \bo{X}_{k'} \bo{\Psi}_{k'}^\H\Big)\bo{H}_{k}^\H\Big]^{-1}\Bigg), \label{eqn:rkc2}
	\end{align}
	where $\bo{Q}_{\ur{c}} = \bo{P}_{\ur{c}} \bo{P}_{\ur{c}}^\H.$ Since the CM symbol vector must be decodeable at all users, its rate is chosen as $R_{\ur{c}} = \bmin{R_{k,\ur{c}}, k=1,\dots,K}.$
				
	\section{WSR Maximization With Perfect CSI Knowledge at the BS}
	\label{sec:mwsr}
	In this section, we formulate the WSR optimization problem for perfect CSI knowledge at the BS, i.e., we assume that the BS knows all matrices $\bo{H}_k,k=1,\dots,K,$ and the scalar $L_k$ perfectly. In addition, we assume that the users know their own MIMO channel matrices perfectly. Since the formulated optimization problem is highly non-convex, we characterize the achievable WSR through a feasible LB obtained via SCA. Lastly, we present a low-complexity alternative to searching over all permutations of user indices. We note that the changes required to accommodate imperfect CSI at the BS are detailed in Section \ref{sec:mwsri}.
		
	\subsection{Weighted Sum Rate}
	\label{sec:wsr}
	
	Let $0 \leq \eta_k \leq 1,k=1,\dots,K,$ $\sum_{k=1}^{K}\eta_k = 1,$ denote fixed weights which can be chosen to adjust the rates of the users during PA \cite[Sec. 4]{WangGiannakis2011}. Furthermore, as described earlier, $\eta_{\ur{c},k}, k=1,\dots,K,$ denote the fraction of the available bits in the CM assigned to user $k.$ Then, the WSR is given by
	\begin{align}
		R_\mathrm{wsr} = \sum_{k=1}^{K} \eta_k\eta_{\ur{c},k}  R_{\ur{c}} + \sum_{k = 1}^{K} \eta_k R_{k}. \label{eqn:wsr}
	\end{align}
	
	\subsection{Problem Formulation}
	\label{sec:probform}
	
	Based on (\ref{eqn:wsr}), the maximum WSR is the solution $R_\mathrm{wsr}^\star$ of the following optimization problem:
	\begin{maxi!}
		{\substack{\bo{Q}_{\ur{c}} \succcurlyeq \bo{0},\\\bo{X}_{k} \succcurlyeq \bo{0},\forall\,k}}
		{R_\mathrm{wsr}}
		{\label{opt:wsr}}
		{R_\mathrm{wsr}^\star =}
		\addConstraint{\text{C1: } \btr{\bo{Q}_{\ur{c}}} + \sum_{k = 1}^{K} \btr{\bo{X}_{k}}\leq P_{\T}\label{cons:t1}}{}{}
		\addConstraint{\text{C2: } \rank{\bo{X}_{k}}\leq M_k,\quad\forall\,k\label{cons:rank2}}{}{}
		\addConstraint{\text{C3: } \rank{\bo{Q}_{\ur{c}}}\leq M,\label{cons:rankc}}{}{}
	\end{maxi!}
	where the WSR is maximized for all permutations of user labels. In the following, we compare $R_\mathrm{wsr}^\star$ in (\ref{opt:wsr}) with the corresponding optima for ZF and BD precoding as well as BD-based MIMO-RSMA.
	
	\begin{remark}
	Based on (\ref{eqn:wsr}), for $\eta_k = \frac{1}{K}, k=1,\dots,K,$ the WSR maximization in (\ref{opt:wsr}) reduces to sum rate (SR) maximization. Hence, WSR maximization is a generalization of SR maximization. Furthermore, the weights $\eta_k, k=1,\dots,K,$ can also be chosen to incorporate quality of service constraints, such as fairness constraints, in the optimization problem, see \cite{WangGiannakis2011} for details.\end{remark}
		
	\subsection{Performance Comparison with ZF and BD Precoding and BD-based MIMO-RSMA}
	\label{sec:bdcomp}
	
	\begin{proposition}
		\label{prop:bdcomp}
		Let $R_\mathrm{wsr}^{\textrm{BD}\star},$ $R_\mathrm{wsr}^{\textrm{BD+MIMO-CM}\star},$ and $R_\mathrm{wsr}^{\textrm{ZF}\star}$ denote the optimal WSRs for BD precoding, BD-based MIMO-RSMA with MIMO CM, and ZF precoding, respectively, analogous to $R_\mathrm{wsr}^\star$ in (\ref{opt:wsr}). Then, we have $R_\mathrm{wsr}^\star \geq R_\mathrm{wsr}^{\textrm{BD+MIMO-CM}\star} \geq R_\mathrm{wsr}^{\textrm{BD}\star}$ and $R_\mathrm{wsr}^\star \geq R_\mathrm{wsr}^{\textrm{ZF}\star}.$
	\end{proposition}
	\begin{proof}
		Please refer to Appendix \ref{app:bdcomp}.
	\end{proof}
	
	\subsection{A Feasible Lower Bound}
	\label{sec:sca}
	\label{sec:ref}
	Solving (\ref{opt:wsr}) entails a very high computational complexity due to the non-convex objective function and the rank constraints in (\ref{cons:rank2}) and (\ref{cons:rankc}). Hence, in the following, we characterize $R_\mathrm{wsr}^\star$ via a feasible LB. To this end, we relax (\ref{opt:wsr}) by eliminating the rank constraints (\ref{cons:rank2}) and (\ref{cons:rankc}) to obtain a relaxed optimization problem \cite{Luo2006} as follows:
	\begin{maxi}
		{\substack{\bo{Q}_{\ur{c}} \succcurlyeq \bo{0},\\\bo{X}_{k} \succcurlyeq \bo{0},\forall\,k}}
		{R_\mathrm{wsr} \quad \text{subject to C1.}}
		{\label{opt:wsr-}}
		{}
	\end{maxi}
		
	However, relaxed optimization problem (\ref{opt:wsr-}) is still non-convex due to the non-convex objective function $R_\mathrm{wsr}.$ Nevertheless, a locally optimal solution of (\ref{opt:wsr-}) can now be obtained via SCA \cite{Razaviyayn2014} based on a first-order approximation of the non-convex terms in $R_\mathrm{wsr}.$ However, $R_\mathrm{wsr}$ is a function of matrix-valued variables. Hence, unlike \cite{Zhou2020, Joudeh2016}, where SCA is utilized for solving optimization problems with scalar or vector variables, solving (\ref{opt:wsr-}) necessitates deriving first-order approximations for scalar functions of positive semi-definite matrix variables. The SCA procedure and the corresponding first-order approximations for solving (\ref{opt:wsr-}) are described in detail in the following. However, as a solution of (\ref{opt:wsr-}) may not satisfy rank constraints  (\ref{cons:rank2}) and (\ref{cons:rankc}), it may not be a feasible solution of (\ref{opt:wsr}). Hence, \emph{based on the obtained locally optimal solution,} we reformulate (\ref{opt:wsr-}) using new optimization variables which satisfy the rank constraints in (\ref{cons:rank2}) and (\ref{cons:rankc}) by construction. The resulting solution for the new optimization problem is a feasible solution for (\ref{opt:wsr}).
	
	\subsubsection{Successive Convex Approximation}
	\label{sec:sca2}
	\begin{figure}
		\begin{algorithm}[H]
			\small
			\begin{algorithmic}[1]
				\STATE {Initialize $\breve{\bo{X}}_{k}^{(0)} = \frac{P_\T}{K N_k}\bo{I}_{N_k},\forall\,k,$ $\tilde{R}_\mathrm{wsr}^{\star(0)} = -\infty,$ numerical tolerance $\epsilon,$ and iteration index $l=0.$}
				\REPEAT	
				\STATE {$l \leftarrow l + 1$}
				\STATE {Update the first-order approximations in (\ref{eqn:tr12}) and (\ref{eqn:trc}) based on $\breve{\bo{X}}_{k}^{(l-1)}.$}
				\STATE {Solve convex optimization problem (\ref{opt:wsrlo}) to obtain the optimal value $\tilde{R}_\mathrm{wsr}^{\star(l)}$ and solution $\bo{Q}_{\ur{c}}^\star, \bo{X}_{k}^\star,\forall\,k.$}
				\STATE {Set $\breve{\bo{X}}_{k}^{(l)} = \bo{X}_{k}^\star,\forall\,k.$}
				\UNTIL {$|\tilde{R}_\mathrm{wsr}^{\star(l)} - \tilde{R}_\mathrm{wsr}^{\star(l-1)}| < \epsilon$}
				\STATE{Return $\tilde{R}_\mathrm{wsr}^\star = \tilde{R}_\mathrm{wsr}^{\star(l)}$ and $\bo{Q}_{\ur{c}}^\star, \bo{X}_{k}^\star,\forall\,k,$ as the optimal value and the corresponding solution.}
			\end{algorithmic}
			\caption{Algorithm for solving (\ref{opt:wsr-}) via SCA.}
			\label{alg:wsrlo}
		\end{algorithm}
	\end{figure}

	For solving (\ref{opt:wsr-}) via SCA, an inner convex optimization problem based on a first-order approximation of the non-convex objective function $R_{\mathrm{wsr}}$ is constructed. This inner optimization problem is solved repeatedly until convergence, up to a numerical tolerance $\epsilon.$ In each iteration, based on the obtained optimal solution, the first-order approximation is updated and used as the objective function for the next iteration. The procedure is described in detail in the following.
	
	In iteration $l=1,2,\dots,$ a convex approximation of the objective function $R_{\mathrm{wsr}},$ denoted by $\tilde{R}_{\mathrm{wsr}},$ is constructed based on a first-order approximation around \emph{given points} $\breve{\bo{X}}_{k}^{(l-1)} \in \mathbb{C}^{N_k\times N_k}, k=1,\dots,K,$ with initial values $\breve{\bo{X}}_{k}^{(0)} = \frac{P_\T}{K N_k}\bo{I}_{N_k}, k=1,\dots,K,$ as follows:
	\begin{align}
		\tilde{R}_\mathrm{wsr} = \sum_{k=1}^{K} \eta_k\eta_{\ur{c},k} \tilde{R}_{\ur{c}} + \sum_{k = 1}^{K} \eta_k \tilde{R}_{k},
	\end{align}
	where $\tilde{R}_{\ur{c}} = \bmin{\tilde{R}_{k,\ur{c}}, k=1,\dots,K},$ and $\tilde{R}_{k,\ur{c}}$ and $\tilde{R}_{k},k=1,\dots,K,$ are given in the following proposition.
	
	\begin{proposition}
		\label{prop:fo}
		Let $h > 0$ and let matrix $\bo{V}_k \in \mathbb{C}^{N_k\times N_k}$ be an arbitrary matrix. Then, the first-order approximations of $R_{k,\ur{c}}$ and $R_k$ at points 
		\begin{align}
			\bo{X}_{k} = \breve{\bo{X}}_{k}^{(l-1)} + h\bo{V}_k, k=1,\dots,K,
		\end{align}
		along direction $\bo{V}_k$ and with sufficiently small $h$ such that $\bo{X}_{k} \succcurlyeq \bo{0}$ is in the neighborhood of $\breve{\bo{X}}_{k}^{(l-1)}$ are given in (\ref{eqn:tr12}) and (\ref{eqn:trc}) on top of the next page,
	\begin{figure*}
	\begin{align}
		\tilde{R}_{k} &= \bloggdet{\bo{I}_{M_k} + \frac{1}{L_k\sigma^2}\bo{H}_k\bo{Q}_{k}\bo{H}_k^\H + \frac{1}{L_k\sigma^2}\bo{H}_k \Big(\sum_{k'=1}^{k-1}\bo{Q}_{k'}\Big)\bo{H}_k^\H} - \bloggdet{\bo{I}_{M_k} + \frac{1}{L_k\sigma^2}\bo{H}_k \breve{\bo{Q}}_k^{(l-1)} \bo{H}_k^\H} \nonumber\\&\quad{}- \frac{1}{L_k\sigma^2\log_{\mathrm{e}}(2)}\sum_{k'=1}^{k-1}\btr{\bo{\Psi}_{k'}^\H \bo{H}_k^\H \Big[\bo{I}_{M_k} + \frac{1}{L_k\sigma^2}\bo{H}_k \breve{\bo{Q}}_k^{(l-1)} \bo{H}_k^\H\Big]^{-1}\bo{H}_k \bo{\Psi}_{k'}\Big(\bo{X}_{k'}-\breve{\bo{X}}_{k'}^{(l-1)}\Big)}, \label{eqn:tr12} \\
		\tilde{R}_{k,\ur{c}} &= \bloggdet{\bo{I}_{M_k} + \frac{1}{L_k\sigma^2}\bo{H}_{k}\bo{Q}_{\ur{c}}\bo{H}_{k}^\H + \frac{1}{L_k\sigma^2}\bo{H}_{k} \Big(\sum_{k'=1}^{k}\bo{Q}_{k'}\Big)\bo{H}_{k}^\H}
		-\bloggdet{\bo{I}_{M_k} + \frac{1}{L_k\sigma^2}\bo{H}_{k} \breve{\bo{Q}}_{k+1}^{(l-1)} \bo{H}_{k}^\H} \nonumber\\&\quad{} - \frac{1}{L_k\sigma^2\log_{\mathrm{e}}(2)}\sum_{k'=1}^{k-1}\btr{\bo{\Psi}_{k'}^\H \bo{H}_k^\H \Big[\bo{I}_{M_k} + \frac{1}{L_k\sigma^2}\bo{H}_k \breve{\bo{Q}}_{k+1}^{(l-1)} \bo{H}_k^\H\Big]^{-1} \bo{H}_k \bo{\Psi}_{k'}\Big(\bo{X}_{k'}-\breve{\bo{X}}_{k'}^{(l-1)}\Big)} \label{eqn:trc}
	\end{align}
	\end{figure*}
	where $\bo{Q}_k = \bo{\Psi}_{k} \bo{X}_{k} \bo{\Psi}_{k}^\H$ and $\breve{\bo{Q}}_{k}^{(l-1)} = \sum_{k'=1}^{k-1}\bo{\Psi}_{k'} \breve{\bo{X}}_{k'}^{(l-1)} \bo{\Psi}_{k'}^\H,$ for $k = 1,\dots,K.$
	\end{proposition}
	\begin{proof}
		Please refer to Appendix \ref{app:fo}.
	\end{proof}
	
	Next, an inner convex optimization problem with $\tilde{R}_{\mathrm{wsr}}$ as the objective function is constructed as follows:
	\begin{maxi}
		{\substack{\bo{Q}_{\ur{c}} \succcurlyeq \bo{0},\\\bo{X}_{k} \succcurlyeq \bo{0},\forall\,k}}
		{\tilde{R}_\mathrm{wsr} \quad\text{subject to C1.}}
		{\label{opt:wsrlo}}
		{\tilde{R}_\mathrm{wsr}^\star =}
	\end{maxi}
	
	The inner convex optimization problem (\ref{opt:wsrlo}) is solved using standard convex optimization tools \cite{Boyd2004} to obtain the optimal value, $\tilde{R}_\mathrm{wsr}^{\star(l)},$ and the corresponding optimal solution $\bo{Q}_{\ur{c}}^\star,\bo{X}_{k}^\star,\forall\,k.$ The obtained solution is used as the starting point for the next iteration, i.e., $\breve{\bo{X}}_{k}^{(l)} = \bo{X}_{k}^\star.$ This process gradually tightens the first-order approximation of the objective function around a local optimum of (\ref{opt:wsr-}). Hence, the corresponding sequence of optimal values of (\ref{opt:wsrlo}), $\tilde{R}_\mathrm{wsr}^{\star(l)},l=1,2,\dots,$ converges to a local optimum of (\ref{opt:wsr-}) \cite{Razaviyayn2014}. The iterations are continued until convergence up to a numerical tolerance $\epsilon.$ The resulting algorithm is summarized in Algorithm \ref{alg:wsrlo}. Furthermore, similarly to (\ref{opt:wsr}), Algorithm  \ref{alg:wsrlo} is applied for all permutations of user labels, and the maximal  $\tilde{R}_\mathrm{wsr}^\star$ and the corresponding solution are exploited for reformulating (\ref{opt:wsr-}), as described below.
		
	\subsubsection{Problem Reformulation}
	\label{sec:reform}

	In the following, we reformulate (\ref{opt:wsr-}) with new optimization variables. Let $\bo{Q}_{\ur{c}}^\star,$ $\bo{X}_{k}^\star,k=1,\dots,K,$ denote the locally optimal solution of (\ref{opt:wsr-}) obtained with Algorithm \ref{alg:wsrlo}. We use this solution to define new matrix structures for $\bo{Q}_{\ur{c}}$ and $\bo{X}_{k},k=1,\dots,K,$ as follows:
	\begin{align}
		\bo{Q}_{\ur{c}} &= \bo{U}_{\ur{c}} \tilde{\bo{X}}_{\ur{c}} \bo{U}_{\ur{c}}^\H, \label{eqn:qkc'} \\
		\bo{X}_{k} &= \bo{U}_{k} \tilde{\bo{X}}_{k} \bo{U}_{k}^\H, \label{eqn:qk'}
	\end{align}
	where $\bo{U}_{\ur{c}} \in \mathbb{C}^{N\times M}$ and $\bo{U}_k \in \mathbb{C}^{N_k\times M_k}$ contain the eigenvectors of $\bo{Q}_{\ur{c}}^\star$ and $\bo{X}_{k}^\star,$ respectively, corresponding to their $M$ and $M_k$ largest eigenvalues, respectively. We follow the convention that the eigenvectors of a zero matrix are zero vectors. Here, $\tilde{\bo{X}}_{\ur{c}} \in \mathbb{C}^{M\times M}$ and $\tilde{\bo{X}}_{k} \in \mathbb{C}^{M_k\times M_k},\forall\,k,$ are symmetric, positive semi-definite matrices which form the new optimization variables.
	
	\begin{remark}
		The rank of $\bo{Q}_{\ur{c}}$ in (\ref{eqn:qkc'}) is limited to $M$ as $\bo{Q}_{\ur{c}}$ is the product of $N\times M,$ $M\times M,$ and $M\times N$ matrices $\bo{U}_{\ur{c}},$ $\tilde{\bo{X}}_{\ur{c}},$ and $\bo{U}_{\ur{c}}^\H,$ respectively. Similarly, the rank of $\bo{X}_{k}$ in (\ref{eqn:qk'}) is limited to $M_k.$
	\end{remark}
		
	Problem (\ref{opt:wsr-}) can now be reformulated in terms of the new matrices as follows:
	\begin{maxi!}
		{\substack{\tilde{\bo{X}}_{\ur{c}} \succcurlyeq \bo{0},\\\tilde{\bo{X}}_{k} \succcurlyeq \bo{0},\forall\,k}}
		{R_\mathrm{wsr}}
		{\label{opt:wsrlb}}
		{\mathrlap{R_\mathrm{wsrlb}^\star =}\nonumber\\}
		\addConstraint{\overline{\text{C1}}\text{: }\btr{\tilde{\bo{X}}_{\ur{c}}} + \sum_{k = 1}^{K} \btr{\tilde{\bo{X}}_{k}} \leq P_{\T},}{}{}
	\end{maxi!}
	where $\overline{\text{C1}}$ is obtained analogously to (\ref{eqn:trsimp}), and the WSR is maximized only for the chosen user index permutation. A locally optimal solution for (\ref{opt:wsrlb}) can be obtained via SCA, i.e., Algorithm \ref{alg:wsrlo} is applicable where the optimization variables in (\ref{opt:wsrlo}) and the first-order approximations in (\ref{eqn:tr12}) and (\ref{eqn:trc}) are revised to account for the new variables and matrix structures in (\ref{eqn:qkc'}) and (\ref{eqn:qk'}).
	
	\begin{remark}
		A locally optimal solution of (\ref{opt:wsrlb}) is a feasible LB for (\ref{opt:wsr}). The solution is feasible because the matrices in (\ref{eqn:qkc'}) and (\ref{eqn:qk'}) ensure rank constraints (\ref{cons:rank2}) and (\ref{cons:rankc}) by construction. However, the solution is suboptimal because the optimization variables in (\ref{opt:wsrlb}) have fewer DoFs compared to those in (\ref{opt:wsr}).
	\end{remark}

	\subsection{Fixed-Permutation Lower Bound}
	\label{sec:subopt}

	We note from the above that $K$ users necessitate a search over $K!$ permutations of user indices. This can be prohibitively complex when there are more than a few users. Hence, in the following, we propose a low-complexity scheme in which only one user permutation is considered. Nevertheless, as shown later in the simulation results in Section \ref{sec:sim}, the performance loss incurred is negligible.
	
	Let $R_k^{\ur{SU}}, k=1,\dots,K,$ denote the single user rates of users $1,\dots,K,$ respectively, given by
	\begin{maxi!}
		{\substack{\bo{P}_{k} \succcurlyeq \bo{0}}}
		{\bloggdet{\bo{I}_{M_k} + \frac{1}{\sigma^2}\bo{H}_k \bo{P}_k \bo{H}_k^\H}}
		{}
		{\mathrlap{R_k^{\ur{SU}} =}\nonumber\\}
		\addConstraint{\btr{\bo{P}_k\bo{P}_k^\H} \leq P_{\T}.}{}{}
	\end{maxi!}
	For the fixed-permutation LB, we optimize (\ref{opt:wsr}) only over the user index permutation $(k_1,k_2,\dots,k_K)$ which satisfies $\eta_{k_1} R_{k_1}^{\ur{SU}} \geq \dots \geq \eta_{k_K} R_{k_K}^{\ur{SU}}.$
	
	\begin{figure*}
		\centering
		\captionof{table}{Complexity of precoder and PA computation for the proposed SNS-based MIMO-RSMA scheme and several baseline schemes.}
		\label{tab:cc}
		\renewcommand{\arraystretch}{1}
		\resizebox{0.75\textwidth}{!}{
			\begin{tabular}{|l|l|l|l|}
				\hline
				\multicolumn{1}{|c|}{\multirow{3}{*}{Scheme}} & \multicolumn{3}{c|}{Complexity} \\\cline{2-4}
				& \multicolumn{1}{c|}{Precoder} & \multicolumn{2}{c|}{PA / Precoder} \\\cline{3-4}
				& &  \multicolumn{1}{c|}{Iterations}  &  \multicolumn{1}{c|}{Total}  \\\hline
				Proposed Scheme\footnotemark & $\bO{2 N^3 + KN\sqrt{N} \log(1/\epsilon)}$ & $N_\ur{iter}$ & $\bO{N_\ur{iter} V\sqrt{V} \log(1/\epsilon)}$ \\
				Direct SCA & $0$ & $N_\ur{iter}$ & $\bO{N_\ur{iter} N^3K\sqrt{K} \log(1/\epsilon)}$  \\
				BD-based MIMO-RSMA & $\bO{2 N^3 + \sum_{k=1}^{K}M_k^3}$ & $N_\ur{iter}$ & $\bO{N_\ur{iter} N^3 \log(1/\epsilon)}$ \\
				RZF & $\bO{\frac{3}{2} N^3}$ & $N_\ur{iter}$ & $\bO{N_\ur{iter} N\sqrt{N} \log(1/\epsilon)}$ \\
				BD & $\bO{2 N^3 + \sum_{k=1}^{K}M_k^3}$ & $0$ & $\bO{N\sqrt{N} \log(1/\epsilon)}$ \\
				ZF & $\bO{\frac{3}{2} N^3}$ & $0$ & $\bO{N\sqrt{N} \log(1/\epsilon)}$ \\\hline
		\end{tabular}}
	\end{figure*}
		
	\subsection{Computational Complexity}
	\label{sec:snscomp}
	In this section, we evaluate the complexity of computing the SNS precoders and compare it with the complexities of the baseline schemes. We assume a critically loaded system and consider the fixed-permutation LB described in the previous section. The resulting computational complexities are as follows.
				
	\begin{itemize}	
	\item For the proposed SNS-based MIMO-RSMA scheme, computing $\bo{\Psi}_k,k=1,\dots,K,$ via the QR decomposition entails a complexity of $\bO{2 N^3}$ \cite{Golub2012}. Furthermore, optimizing $\bo{P}_\ur{c}$ and $\bo{X}_k,k=1,\dots,K,$ which involves $\bO{V}$ real-valued optimization variables, $V = N^2 + \sum_{k=1}^{K} N_k^2,$ via SCA with $N_\ur{iter}$ iterations and numerical tolerance $\epsilon$ entails a complexity of $\bO{N_\ur{iter} V\sqrt{V} \log(1/\epsilon)}$ \cite{Wright1997,Serrano2015}. Moreover, computing the user index permutation (given in Section \ref{sec:subopt}) entails an additional complexity of $\bO{K N\sqrt{N} \log(1/\epsilon)}$ for solving $K$ convex optimization problems involving $N^2$ real-valued optimization variables, respectively.
		
	\item Directly optimizing the precoding matrices $\bo{P}_\ur{c}, \bo{P}_k, k=1,\dots,K,$ via SCA involves $\bO{KN^2}$ real-valued variables and entails an overall complexity of $\bO{N_\ur{iter} N^3K\sqrt{K} \log(1/\epsilon)}$ \cite{Wright1997,Serrano2015}.

	\item For BD-based MIMO-RSMA, computing the BD precoders via QR decomposition and SVD entails a complexity of $\bO{2 N^3 + \sum_{k=1}^{K}M_k^3}$ \cite{Golub2012}. Furthermore, PA via SCA involving $\bO{N^2}$ real-valued variables entails an additional complexity of $\bO{N_\ur{iter} N^3 \log(1/\epsilon)}$ \cite{Wright1997,Serrano2015}.	

	\item For RZF precoding, computing the precoder via matrix inversion and multiplication entails a complexity of $\bO{\frac{3}{2} N^3}$ \cite{Krishnamoorthy2013}, and PA via SCA involving $\bO{N}$ optimization variables entails a complexity of $\bO{N_\ur{iter} N\sqrt{N} \log(1/\epsilon)}$ \cite{Wright1997,Serrano2015}.
	
	\item For BD and ZF precoding, computing the precoders entail complexities of $\mathcal{O}\big(2 N^3$ $+ \sum_{k=1}^{K}M_k^3\big)$ \cite{Golub2012} and $\bO{\frac{3}{2} N^3}$ \cite{Krishnamoorthy2013}, respectively, and finding the corresponding PA via convex optimization entails a complexity of $\bO{N\sqrt{N} \log(1/\epsilon)}$ \cite{Wright1997,Serrano2015}.
	\end{itemize}

	The computational complexities of the different schemes are summarized in Table \ref{tab:cc}, shown on top of the next page. From the table, we note that direct optimization of the precoders entails the highest overall complexity, followed by (in order) the proposed SNS-based MIMO-RSMA, BD-based MIMO-RSMA, RZF precoding, BD precoding, and ZF precoding.
		
	\begin{remark}
		We note that although the proposed Algorithm \ref{alg:wsrlo} has a comparatively high computational complexity, it facilitates the performance evaluation of SNS-based MIMO-RSMA and provides a benchmark for the development of low-complexity algorithms, which might be developed in the future.
	\end{remark}
			
	\section{WSR Maximization With Imperfect CSI Knowledge at the BS}
	\label{sec:mwsri}
	In this section, we consider WSR maximization with imperfect CSI knowledge at the BS. We begin by describing the changes to SNS precoding and decoding provided in Section \ref{sec:pd}. Next, we study the robustness of SNS precoding via derivative-based SA. Lastly, we present the changes required to the WSR maximization algorithm proposed in Section \ref{sec:mwsr} if the CSI is imperfect.
	\footnotetext{For the proposed scheme, the precoder column includes the complexity incurred by the required user index permutation selection.}
	
	\subsection{Changes to Precoding and Decoding}
	As only the imperfect MIMO channel matrices of the users, $\bar{\bo{H}}_k,k=1,\dots,K,$ are available at the BS, the SNS precoders are chosen based on these imperfect matrices as $\bo{P}_k = \bar{\bo{\Psi}}_k \bo{X}_k.$ Here, for user $k,$ $\bar{\bo{\Psi}}_k \in \mathbb{C}^{N\times N_k}$ denotes the matrix whose columns contain the unit-length basis vectors of the null space of the augmented matrix $$\bar{\bo{F}}_k = \begin{bmatrix}
			\bar{\bo{H}}_{1}^\T & \bar{\bo{H}}_{2}^\T & \dots & \bar{\bo{H}}_{k-1}^\T
		\end{bmatrix}^\T\negthickspace$$ with the convention $\bar{\bo{\Psi}}_{1} = \bo{I}_N,$ analogous to the precoding scheme in Section \ref{sec:precoding}. However, since the precoders are computed based on $\bar{\bo{H}}_k,$ the orthogonality properties with respect to $\bo{H}_k$ described in Section \ref{sec:precoding} may not hold. Hence, in this case, the signal received at user $k$ suffers from interference by all other users $k'\neq k,k'=1,\dots,K,$ i.e., $\bo{y}_k$ in (\ref{eqn:y_r}) is replaced by
	\begin{align}
		\bo{y}_k &= \frac{1}{\sqrt{\mathstrut L_k}} \bo{H}_k\bo{P}_{\ur{c}}\bo{s}_{\ur{c}} + \frac{1}{\sqrt{\mathstrut L_k}} \bo{H}_k \sum_{k'=1}^{K} \bar{\bo{\Psi}}_{k'} \bo{X}_{k'}^\frac{1}{2}\bo{s}_{k'} + \bo{z}_k. \label{eqn:y_ri}
	\end{align}
	Decoding at user $k$ is carried out as described in Section \ref{sec:decoding}.
	
	\subsection{IUI due to Imperfect CSI at the BS}
	\label{sec:iuiimp}

	For given $\bo{X}_k,k=1,\dots,K,$ let
	\begin{align}
		\bo{\Xi}^\uparrow_k &= \frac{1}{\sqrt{\mathstrut L_k}} \bo{H}_k \sum_{k'=k+1}^{K} \bar{\bo{\Psi}}_{k'} \bo{X}_{k'}^\frac{1}{2}\bo{s}_{k'} \label{eqn:i_k}
	\end{align}
	denote the additional IUI from users $k+1,\dots,K$ due to the imperfect estimate of $\bo{H}_k$ and let
	\begin{align}
		\bo{\Xi}^\downarrow_k &{}= \frac{1}{\sqrt{\mathstrut L_k}} \bo{H}_k \sum_{k'=1}^{k-1} (\bar{\bo{\Psi}}_{k'} - \bo{\Psi}_{k'}) \bo{X}_{k'}^\frac{1}{2}\bo{s}_{k'} \nonumber\\&{}= \frac{1}{\sqrt{\mathstrut L_k}} \bo{H}_k \sum_{k'=1}^{k-1} \Delta \bo{\Psi}_{k'} \bo{X}_{k'}^\frac{1}{2}\bo{s}_{k'}, \label{eqn:j_k}
	\end{align}
	where $\Delta \bo{\Psi}_{k} = \bar{\bo{\Psi}}_{k} - \bo{\Psi}_{k},$ denote the additional IUI from users $1,\dots,k-1$ due the imperfect estimates of $\bo{H}_{k},k=1,\dots,k-1.$ Both $\bo{\Xi}^\uparrow_k$ and $\bo{\Xi}^\downarrow_k,$ which are present in (\ref{eqn:y_ri}) but absent in (\ref{eqn:y_r}), are characterized in the following proposition.
		
	\begin{proposition}
		\label{prop:ikjk}
		Matrices $\bar{\bo{\Psi}}_{k}, \bo{\Psi}_{k},k=1,\dots,K,$ can be chosen such that:
		\begin{align}
			\|\bo{\Xi}^\uparrow_k\| &\leq  \frac{1}{\sqrt{\mathstrut L_k}} \|\Delta\bo{H}_k\|  \sum_{k'=k+1}^{K} \|\bo{X}_{k'}^\frac{1}{2}\bo{s}_{k'}\|, \label{eqn:normik} \\
			\|\bo{\Xi}^\downarrow_k\| &\leq  \frac{1}{\sqrt{\mathstrut L_k}} \|\bo{H}_k\| \nonumber\\&\quad\times\sum_{k'=1}^{k-1} \Big[\|(\bar{\bo{F}}_{k'})^+\| \|\Delta\bo{F}_{k'}\| - \log(1-\|\bo{C}_{k'}\|)\Big] \nonumber\\&\quad\times \|\bo{X}_{k'}^\frac{1}{2}\bo{s}_{k'}\|, \label{eqn:normjk}
		\end{align}
	where $\Delta\bo{F}_{k} = \bar{\bo{F}}_{k} - \bo{F}_{k}$ and $\bo{C}_{k} = \bo{\Psi}_{k}^\H(\Delta\bo{F}_{k})^\H\big(\bar{\bo{F}}_{k}\bar{\bo{F}}_{k}^\H\big)^{-1}\Delta\bo{F}_{k}\bo{\Psi}_{k}.$
	\end{proposition}
	\begin{proof}
		Please refer to Appendix \ref{app:ikjk}.
	\end{proof}
	
	From Proposition \ref{prop:ikjk}, we observe that $\|\bo{\Xi}^\uparrow_k\|$ depends on the imperfection of user $k$'s own CSI, whereas $\|\bo{\Xi}^\downarrow_k\|$ depends on the imperfections of the other users' CSI. Furthermore, we note that for both $\|\bo{\Xi}^\uparrow_k\|$ and $\|\bo{\Xi}^\downarrow_k\|$ small imperfections of the MIMO channel matrices of the users only cause a small change in the IUI. Hence, for SNS precoding, a user with high quality CSI experiences a \emph{small amount of additional IUI $\|\bo{\Xi}^\uparrow_k\|$ and causes a small amount of additional IUI $\|\bo{\Xi}^\downarrow_k\|$ to others,} thereby ensuring the overall performance and robustness.

	\begin{remark}
		While Proposition \ref{prop:ikjk} only proves the existence of such matrices $\bar{\bo{\Psi}}_{k}, \bo{\Psi}_{k},k=1,\dots,K,$ in Section \ref{sec:sim}, we show that the bounds hold for Gaussian MIMO channels, where $\bar{\bo{\Psi}}_{k}, \bo{\Psi}_{k},k=1,\dots,K,$ are computed via the Gram-Schmidt procedure \cite{Golub2012}.
	\end{remark}
	
	\subsection{Achievable Rate and WSR}
	\label{sec:wsri}
	
	Next, based on the received signal given in (\ref{eqn:y_ri}), the achievable rate is given as follows.
	The rate of the private symbol vector $\bo{s}_{k}$ is given by
	\begin{align}
		\bar{R}_{k} &= \log_2\det\Bigg(\bo{I}_{M_k} + \frac{1}{L_k}\bo{H}_k\bar{\bo{Q}}_{k}\bo{H}_k^\H \nonumber\\&\quad\times\Big[\sigma^2\bo{I}_{M_k} + \frac{1}{L_k}\bo{H}_k \Big(\sum_{k'=1,k'\neq k}^{K}\bar{\bo{Q}}_{k'}\Big)\bo{H}_k^\H\Big]^{-1}\Bigg), \label{eqn:rk2'}
	\end{align}
	where $\bar{\bo{Q}}_{k} = \bar{\bo{\Psi}}_{k} \bo{X}_{k} \bar{\bo{\Psi}}_{k}^\H, k=1,\dots,K.$ Next, for the common symbol vector, the achievable rate at user $k$ is:
	\begin{align}
		\bar{R}_{k,\ur{c}} &= \log_2\det\Bigg(\bo{I}_{M_k} + \frac{1}{L_k}\bo{H}_{k}\bo{Q}_{\ur{c}}\bo{H}_{k}^\H \nonumber\\&\quad\times\Big[\sigma^2\bo{I}_{M_{k}} + \frac{1}{L_k}\bo{H}_{k} \Big(\sum_{k'=1}^{K}\bar{\bo{Q}}_{k'}\Big)\bo{H}_{k}^\H\Big]^{-1}\Bigg), \label{eqn:rkc'}
	\end{align}
	where, as earlier, $\bar{R}_{\ur{c}} = \bmin{\bar{R}_{k,\ur{c}}, k=1,\dots,K}.$ Here, $\bo{Q}_{\ur{c}}$ and $\bo{X}_k,k=1,\dots,K,$ are as defined in Section \ref{sec:mwsr}. Furthermore, the WSR is defined analogously to the definition in (\ref{eqn:wsr}) as follows:
	\begin{align}
		\bar{R}_\mathrm{wsr} = \sum_{k=1}^{K} \eta_k\eta_{\ur{c},k} \bar{R}_{\ur{c}} + \sum_{k = 1}^{K} \eta_k \bar{R}_{k}. \label{eqn:wsr'}
	\end{align}
			
	Now, we describe the changes required for the WSR maximization algorithm in Section \ref{sec:mwsr} for imperfect CSI.

	\subsection{Changes to WSR Maximization Algorithm for Imperfect CSI}
	\label{sec:wsrmi}

	As the BS does not have knowledge of $\bo{H}_k,k=1,\dots,K,$ the true WSR, $\bar{R}_\mathrm{wsr}$ in (\ref{eqn:wsr'}), which depends on $\bo{H}_k,$ cannot be utilized as the objective function for optimizing $\bo{Q}_{\ur{c}}$ and $\bo{X}_k,k=1,\dots,K,$ at the BS. Instead, we propose the use of the objective function:
	\begin{align}
		\bar{\bar{R}}_\mathrm{wsr} &= \sum_{k=1}^{K} \eta_k\eta_{\ur{c},k} \bar{R}_{\ur{c}} \;\Big|_{\bo{H}_k \to \bar{\bo{H}}_k, k=1,\dots,K} \nonumber\\&\quad+ \sum_{k = 1}^{K} \eta_k \bar{R}_{k} \;\Big|_{\bo{H}_k \to \bar{\bo{H}}_k, k=1,\dots,K}, \label{eqn:wsr''}
	\end{align}
	in which $\bo{H}_k$ is replaced by $\bar{\bo{H}}_k,k=1,\dots,K.$ Now, $\bo{Q}_{\ur{c}}$ and $\bo{X}_k,k=1,\dots,K,$ can be optimized at the BS via SCA analogously to the method in Section \ref{sec:mwsr}, i.e., Algorithm \ref{alg:wsrlo} is applicable where the objective function in (\ref{opt:wsrlo}) is updated based on (\ref{eqn:wsr''}).
			
	\section{Simulation Results}
	\label{sec:sim}	
	In this section, we first validate the analytical UBs presented in Proposition \ref{prop:ikjk} and study the computational complexities in Section \ref{sec:snscomp} in terms of the system dimension. Next, we examine the convergence rate of Algorithm \ref{alg:wsrlo} and the optimized beam patterns for SNS and BD precoding. Then, we evaluate the performance of the proposed SNS-based MIMO-RSMA scheme and compare it with the performances of the considered baseline schemes for perfect and imperfect CSI. 
	
	For our simulations, unless specified otherwise, we model the elements of the MIMO channel matrices of the users, which capture the small-scale fading effects, as independent and identically distributed (i.i.d.) Gaussian random variables\footnote{We note that the proposed SNS-based MIMO-RSMA scheme is applicable for arbitrary MIMO channel models \cite{Raghavan2017, Jaeckel2014}, see Figures \ref{fig:wsr-1-1-2-2-4-4-250etc-eta-th} and \ref{fig:wsr-1-1-2-2-4-4-250etc-eta-sik-th}.} $[\boldsymbol{H}_k]_{ij} \sim \mathcal{CN}(0,1),\forall\,i,j,k.$ The scalar path loss for user $k, L_k,$ is set to $d_k^2,$ where $d_k$ denotes the distance (in meters) of user $k$ from the BS. 
	
	The elements of $\Delta\bo{H}_k \sim \mathbb{C}^{M_k\times N}, k=1,\dots,K,$ are also modeled as i.i.d. Gaussian random variables $[\Delta\bo{H}_k]_{ij} \sim \mathcal{CN}(0,\mu_k),\forall\,i,j,k.$ Here, $\mu_k$ models the CSI error variance of user $k.$ Furthermore, we assume that $[\Delta\bo{H}_k]_{ij}$ and $[\bo{H}_k]_{ij},\forall\,i,j,k,$ are statistically independent. Let $\bo{H}_k^{(n)}$ and $\Delta\bo{H}_k^{(n)}$ denote the $n$-th realization of the MIMO channel matrix and the error matrix for user $k,$ respectively. Then, the $n$-th realization of the imperfect MIMO channel matrices of the user $k$ is given by $\frac{1}{\sqrt{\mathstrut L_k}} \bar{\bo{H}}_k^{(n)} = \frac{1}{\sqrt{\mathstrut L_k}} \bo{H}_k^{(n)} + \frac{1}{\sqrt{\mathstrut L_k}} \Delta\bo{H}_k^{(n)}, \quad k=1,\dots,K.$
	
	\subsection{Upper Bounds and Computational Complexities}
	
	\begin{figure*}
		\centering
		\begin{minipage}[t]{0.48\textwidth}
			\centering
			\includegraphics[width=0.8\textwidth]{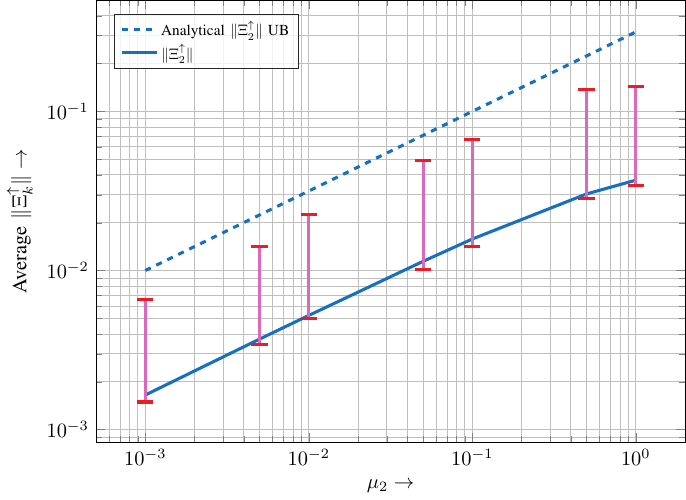}
			\caption{Average $\|\bo{\Xi}^\uparrow_k\|$ with $K=6,$ $N=12,$ $M_k=2,\forall\,k,$ as a function of $\mu_2$ with $d_k = 50\text{ m,}\,\forall\,k,$ and $\mu_k= 0,\forall\,k\neq 2.$ Error bars show the range of values encountered.}
			\label{fig:ik}
		\end{minipage}%
		\hspace{0.02\textwidth}%
		\begin{minipage}[t]{0.48\textwidth}
			\centering
			\includegraphics[width=0.8\textwidth]{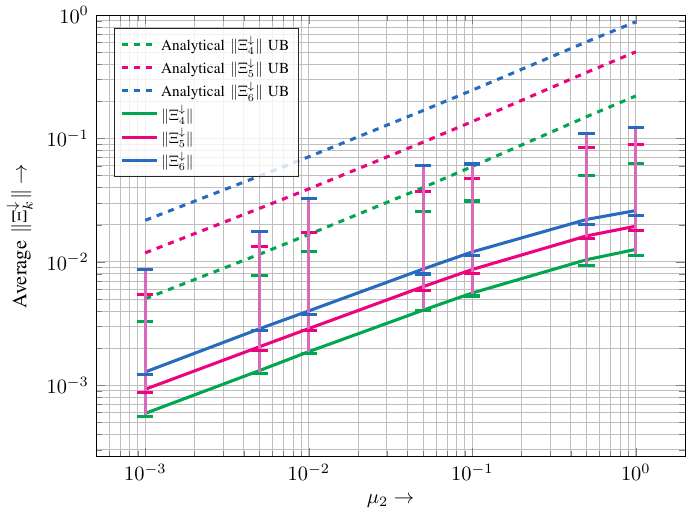}
			\caption{Average $\|\bo{\Xi}^\downarrow_k\|$ with $K=6,$ $N=12,$ $M_k=2,\forall\,k,$ as a function of $\mu_2$ with $d_k = 50\text{ m,}\,\forall\,k,$ and $\mu_k= 0,\forall\,k\neq 2.$ Error bars show the range of values encountered.}
			\label{fig:jk}
		\end{minipage}%
	\end{figure*}
	
	First, we validate the analytical UBs presented in Proposition \ref{prop:ikjk}. For our analysis, we set $K=6, N=12,$ $M_k=2,\forall\,k,$ $d_k = 50 \text{ m},$ and $\|\bo{X}_{k'}^\frac{1}{2}\bo{s}_{k'}\|=1,\forall\,k.$  $\bar{\bo{\Psi}}_{k}, \bo{\Psi}_{k},k=1,\dots,K,$ are computed via the Gram-Schmidt procedure \cite{Golub2012}. We assume the BS to have perfect CSI for all users except for user $2,$ i.e., $\mu_k = 0, k\neq 2,$ and study the impact of the imperfect CSI of user $2$ on the IUI via  $\|\bo{\Xi}^\uparrow_k\|$ and $\|\bo{\Xi}^\downarrow_k\|, k=1,\dots,K.$ Note that, for the given setup, $\|\bo{\Xi}^\uparrow_k\| = 0$ for $k\neq 2,$ and $\|\bo{\Xi}^\downarrow_k\| = 0$ for $k=1,2,3.$
	
	In Figures \ref{fig:ik} and \ref{fig:jk}, $\|\bo{\Xi}^\uparrow_k\|$ and $\|\bo{\Xi}^\downarrow_k\|,$ averaged over $10^4$ MIMO channel realizations, are plotted as a function of $\mu_2.$ Furthermore, the error bars show the range of values of $\|\bo{\Xi}^\uparrow_k\|$ and $\|\bo{\Xi}^\downarrow_k\|$ encountered during the simulations. From the figures, we note that the analytical UBs developed in Proposition \ref{prop:ikjk} accurately bound $\|\bo{\Xi}^\uparrow_k\|$ and $\|\bo{\Xi}^\downarrow_k\|.$ Nevertheless, for high $\mu_k,$ the UB in (\ref{eqn:normjk}) becomes loose as it violates the assumption that $\|\bo{C}\| \ll 1$ in (\ref{eqn:dnnbound}). Furthermore, we observe that the average $\|\bo{\Xi}^\uparrow_k\|$ and $\|\bo{\Xi}^\downarrow_k\|$ are close to the minimum indicating that our choice $\bo{Z} = -\bo{I}_{n-m}$ in (\ref{eqn:sol1}) is suboptimal in most cases. Lastly, we observe a linear dependence between  $\mu_2$ and the resulting IUI, thereby affirming the robustness of SNS precoding, i.e., that small imperfections in CSI cause a small amount of additional IUI.
	
	\begin{figure*}
		\centering
		\begin{minipage}{0.48\textwidth}
			\centering
			\includegraphics[width=0.8\textwidth]{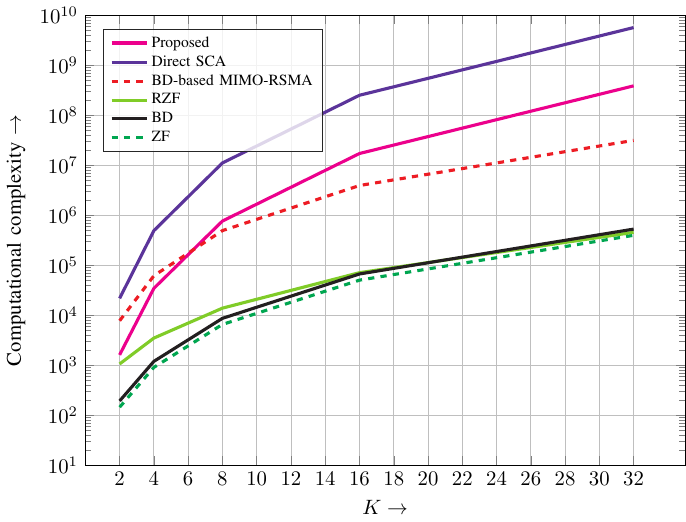}
			\caption{Comparison of complexities of precoder computation and PA for the schemes included in Table \ref{tab:cc} as a function of $K$ for $N=2K,$ $M_k=2,\forall\,k.$}			
			\label{fig:complex}
		\end{minipage}%
		\hspace{0.02\textwidth}%
		\begin{minipage}{0.48\textwidth}
			\centering
			\includegraphics[width=0.8\textwidth]{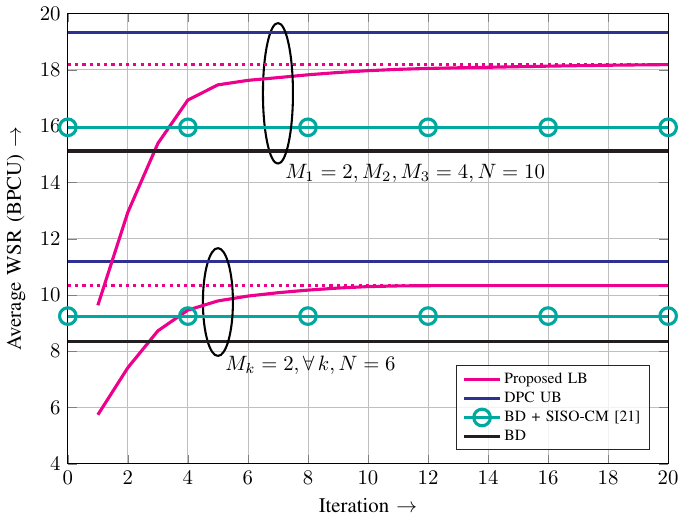}
			\caption{Convergence of Algorithm \ref{alg:wsrlo} with perfect CSI knowledge for $K=3,$ $P_\T = 20 \text{ dBm,}$ $d_k = 50\text{ m,}\,\forall\,k,$ and $\eta_k = \frac{1}{K},\forall\,k.$}
			\label{fig:c-2-4-4-50}
		\end{minipage}%
	\end{figure*}

	Next, in Figure \ref{fig:complex}, we study the complexity of precoder computation and PA for the schemes included in Table \ref{tab:cc} as a function of $K.$ We assume two-antenna users and a critically loaded system with $N = 2K.$ For the iterative schemes in Table \ref{tab:cc}, $N_\mathrm{iter}$ is chosen appropriately to ensure convergence; and the numerical tolerance is set to $\epsilon = 10^{-5}.$ As a measure of complexity, we plot the argument of function $\bO{\cdot}$ given in Table \ref{tab:cc}. From the figure, we observe that `Direct SCA' incurs the largest complexity as it requires solving an optimization problem involving $KN^2$ matrix valued optimization variables. On the other hand, BD, ZF, and RZF entail the lowest complexity among the considered schemes. The proposed precoding scheme incurs an order-of-magnitude lower complexity than `Direct SCA' and, for $K \leq 6,$ a lower complexity than BD-based MIMO-RSMA. For $K > 6,$ the complexity of the proposed SNS-based MIMO-RSMA is higher than that of BD-based MIMO-RSMA.
	
	\subsection{Weighted Sum Rate}

	In this section, we first examine the convergence rate of Algorithm \ref{alg:wsrlo} and the optimized beam patterns for SNS and BD precoding. Next, we compare the performance of SNS-based MIMO-RSMA with that of the baseline schemes. For the following results, the general and fixed-permutation (F-PERM) LBs on the WSR of the proposed SNS-based MIMO-RSMA scheme are obtained as described in Section \ref{sec:mwsr}. An UB for the WSR is obtained based on DPC by exploiting the BC-MAC duality \cite{Vishwanath2003} and \cite[Theorem 1]{Liu2008}. Owing to the difficulty of finding the global optimum of the non-convex rank-constrained optimization problem in (\ref{opt:wsr}), we utilize the DPC UB as a performance upper bound for the WSR of SNS-based MIMO-RSMA. ZF, RZF, and BD precoding are as in \cite{Wiesel2008}, \cite{Peel2005}, and \cite{Spencer2004}, respectively. For BD-based MIMO-RSMA, the scheme in  \cite{Flores2019} (BD+SISO-CM), and the scheme presented in Section \ref{sec:bdrsma} (BD+MIMO-CM) are considered. Lastly, `Direct SCA LB' denotes a LB on the WSR obtained by directly optimizing $\bo{P}_\ur{c},$ $\bo{P}_k,k=1,\dots,K,$ via SCA with zeros as initial values.

	For the general and fixed-permutation LBs, SCA is utilized to solve the WSR maximization problem, as described in Section \ref{sec:sca}. For the DPC UB, ZF precoding, and BD  precoding, the PA problems for WSR maximization are convex optimization problems. For BD+MIMO-CM, the PA problem for WSR maximization is convex for equal user weights $\eta_k,k=1,\dots,K,$ and when the non-convex rank constraint C3 in (\ref{cons:rankc}) is eliminated. In this case, we denote the UB on the maximum WSR by BD+MIMO-CM UB. Otherwise, and for RZF and BD+SISO-CM, an SCA-based solution with zeros as initial values is used.
	
	For the case with imperfect CSI at the BS, as an UB is unknown \cite{Gamal2011}, again the DPC UB for perfect CSI is utilized. The general and fixed-permutation LBs on the WSR for the proposed SNS-based MIMO-RSMA scheme are obtained as described in Section \ref{sec:wsrmi}. For the other schemes, WSR optimization is carried out using the estimated MIMO channel matrices of the users, and the WSR is computed based on the obtained optimized precoders and actual MIMO channel matrices, analogous to the procedure used for SNS-based MIMO-RSMA.
	
	The obtained WSRs are averaged over multiple MIMO channel realizations, for both perfect and imperfect CSI, so as to obtain a $99\%$ confidence interval of $\pm 1$ bit per channel use (BPCU) for the average WSRs. Lastly, for our simulations, we utilize $\eta_{c,k} = \eta_k,$ noise variance $\sigma^2 = -35 \text{ dBm,}$ and, in Algorithm \ref{alg:wsrlo}, numerical tolerance $\epsilon = 10^{-5}.$ 
	
	First, in Figure \ref{fig:c-2-4-4-50}, we study the convergence of Algorithm \ref{alg:wsrlo} for $K=3,$ $N=6,$ $M_k=2,\forall,k,$ and $N=10,$ $M_k=(2,4,4).$ The transmit power budget is set to $P_\T = 20 \text{ dBm.}$ The user distances from the BS are set to $d_k=50 \text{ m},k=1,2,3.$ From the figure, we observe that, for the considered cases, the WSR converges in about $20$ iterations. Furthermore, comparing the convergence speed for the two different antenna configurations, we note that the convergence speed decreases only moderately as the number of variables increases.
	
	\begin{figure*}
		\centering
		\begin{minipage}[t]{0.48\textwidth}
			\centering
			\includegraphics[width=0.8\textwidth]{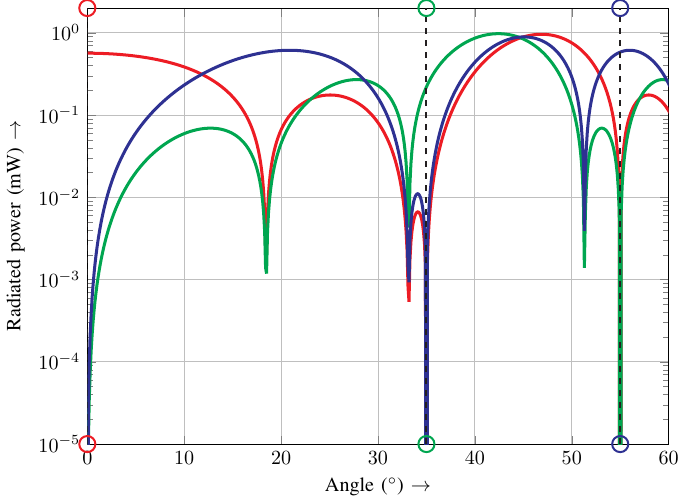}
			\caption{Beam pattern of the optimized BD precoding vectors for $K=4, N=4,$ $M_k=1,\forall\,k,$ $d_k = 50\text{ m,}\,\forall\,k,$ $P_\T = 20 \text{ dBm,}$ and $\eta_k = \frac{1}{K},\forall\,k.$}
			\label{fig:pbd}
		\end{minipage}%
		\hspace{0.02\textwidth}%
		\begin{minipage}[t]{0.48\textwidth}
			\centering
			\includegraphics[width=0.8\textwidth]{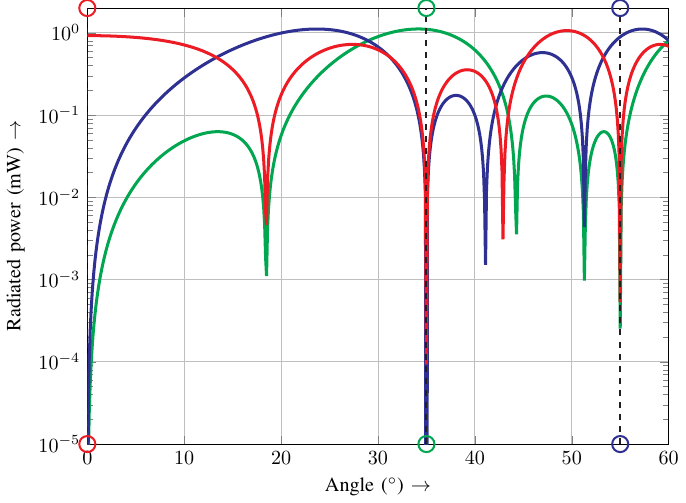}
			\caption{Beam pattern of the optimized SNS precoding vectors for $K=4, N=4,$ $M_k=1,\forall\,k,$ $d_k = 50\text{ m,}\,\forall\,k,$ $P_\T = 20 \text{ dBm,}$ and $\eta_k = \frac{1}{K},\forall\,k.$}
			\label{fig:psns}
		\end{minipage}%
	\end{figure*}

	Next, in Figures \ref{fig:pbd} and \ref{fig:psns}, we show the beam patterns of the BD and SNS precoding vectors based on the parametric MIMO channel model in \cite{Raghavan2017}. To this end, we consider a critically loaded system with $K=4, N=4,$ $M_k=1,\forall\,k,$ $d_k = 50\text{ m,}\,\forall\,k,$ $P_\T = 20 \text{ dBm,}$ and $\eta_k = \frac{1}{K},\forall\,k,$ with the four users located at angles $0^\circ, 35^\circ, 55^\circ,$ and $85^\circ,$ and plot the radiated power as a function of the transmit angle. From Figure \ref{fig:pbd}, we note that, for BD, as expected, the beam nulls are located such that IUI is completely eliminated. On the other hand, in Figure \ref{fig:psns}, for SNS precoding, we observe a residual interference of $10^{-4}$ mW from user $1$ to user $2$ and a residual interference of about $10^{-3}$ mW to $10^{-4}$ mW from users $1$ and $2$ to user $3.$ Moreover, we note that the power radiated towards the users is higher for SNS precoding compared to BD precoding due to the increased flexibility obtained by allowing controlled interference, resulting in a better performance. Furthermore, for BD precoding, as IUI is completely eliminated by design, a minor beam misalignment, e.g., due to imperfect CSI, causes a large additional IUI and poor performance. However, for SNS precoding, since users $2$ and $3$ already experience residual IUI, the impact of a minor beam misalignment is lower, thereby resulting in an enhanced robustness.
	
	\begin{figure*}
		\centering
		\begin{minipage}[t]{0.48\textwidth}
			\centering
			\includegraphics[width=0.8\textwidth]{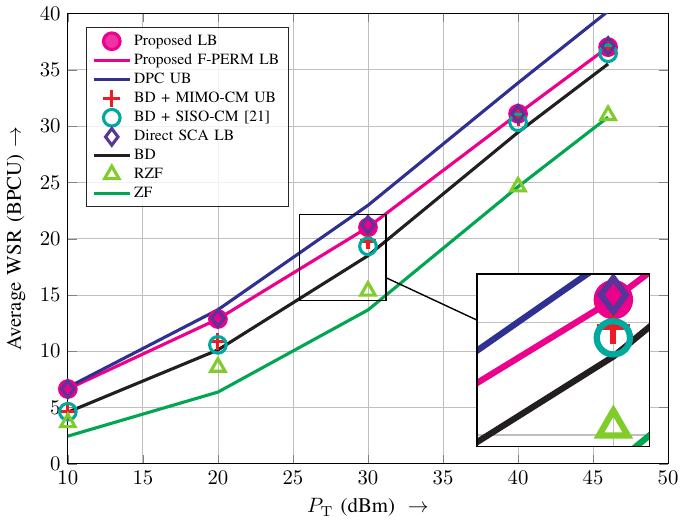}
			\caption{Average WSR for perfect CSI knowledge as a function of $P_\T$ for $K=3,$ $N=10,$ $M_k=(2,4,4),$  $d_k = (250,150,50)\text{ m,}$ and $\eta_k = \frac{1}{K},\forall\,k.$}
			\label{fig:wsr-2-4-4-250etc}
		\end{minipage}%
		\hspace{0.02\textwidth}%
		\begin{minipage}[t]{0.48\textwidth}
			\centering
			\includegraphics[width=0.8\textwidth]{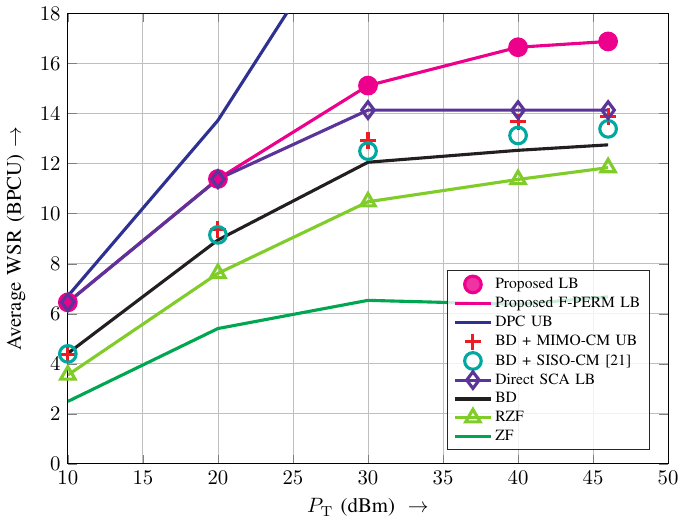}
			\caption{Average WSR for imperfect CSI knowledge as a function of $P_\T$ for $K=3,$ $N=10,$ $M_k=(2,4,4),$  $d_k = (250,150,50)\text{ m,}$ $\eta_k = \frac{1}{K},\forall\,k,$ and $\mu_k= (0.5,0.1,0.01).$}
			\label{fig:wsr-2-4-4-250etc-sik}
		\end{minipage}
	\end{figure*}
	
	Now, we study the ergodic WSR in detail based on the i.i.d. Gaussian MIMO channel model described earlier. In Figures \ref{fig:wsr-2-4-4-250etc} and  \ref{fig:wsr-2-4-4-250etc-sik}, we consider a critically loaded system with $K=3, N=10,$ and $M_k=(2,4,4),$ for perfect and imperfect CSI knowledge at the BS, respectively. The user distances are set to $d_k = (250,150,50)\text{ m,}$ and for the imperfect CSI, $\mu_k= (0.5,0.1,0.01).$ The user weights are set to be equal. From Figure \ref{fig:wsr-2-4-4-250etc}, we note that the proposed LB outperforms RZF, ZF, and BD precoding, and BD-based MIMO-RSMA. This is because, unlike BD precoding, SNS precoding exploits the unused DoFs of the users with lower rates and the DoFs of the users with (near) orthogonal MIMO channels matrices to enhance performance. Furthermore, ZF precoding has a poor performance in the considered critically loaded system as it requires the inversion of an ill-conditioned matrix, see \cite{Krishnamoorthy2020} for more details. In addition, for perfect CSI and the noise-limited low SNR regime, the performance of SNS-based MIMO-RSMA is close to that of DPC. However, in the interference-limited high SNR regime, the performance of the proposed LB approaches that of BD+MIMO-CM UB. This is because, in an interference limited system, the best strategy is to decode all substantial interference. Moreover, although the proposed LB and the Direct SCA LB have negligible performance difference for perfect CSI, from Figure \ref{fig:wsr-2-4-4-250etc-sik}, we observe that, for imperfect CSI, SNS-based MIMO-RSMA outperforms Direct SCA LB and the considered baseline schemes owing to the robustness of SNS precoding, see Section \ref{sec:iuiimp} for details. In this case, we also observe that the performance loss due to fixing the user index permutations (F-PERM) is negligible on average.
	
	\begin{figure*}
		\centering
		\begin{minipage}[t]{0.48\textwidth}
			\centering
			\includegraphics[width=0.8\textwidth]{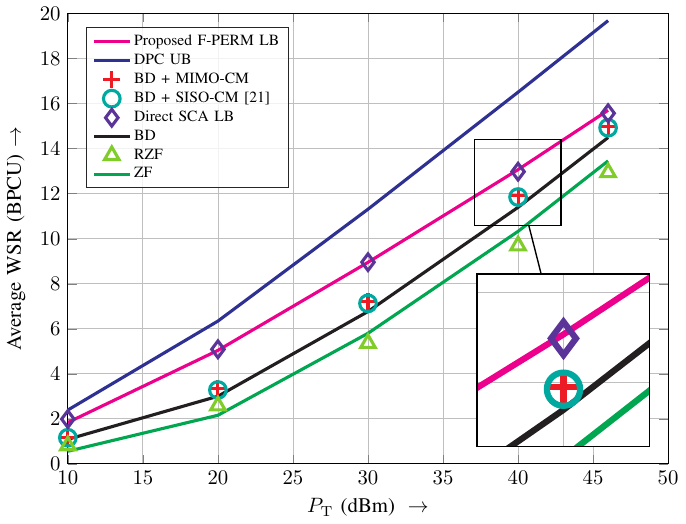}
			\caption{Average WSR for perfect CSI knowledge as a function of $P_\T$ for $K=6,$ $N=14,$ $M_k=(1,1,2,2,\\4,4),$ $d_k = (250,250,150,150,50,50)\text{ m,}$ and $\eta_k = (0.3, 0.3, 0.15, 0.15, 0.05, 0.05),$ and the i.i.d. Gaussian MIMO channel model.}
			\label{fig:wsr-1-1-2-2-4-4-250etc-eta}
		\end{minipage}%
		\hspace{0.02\textwidth}%
		\begin{minipage}[t]{0.48\textwidth}
			\centering
			\includegraphics[width=0.8\textwidth]{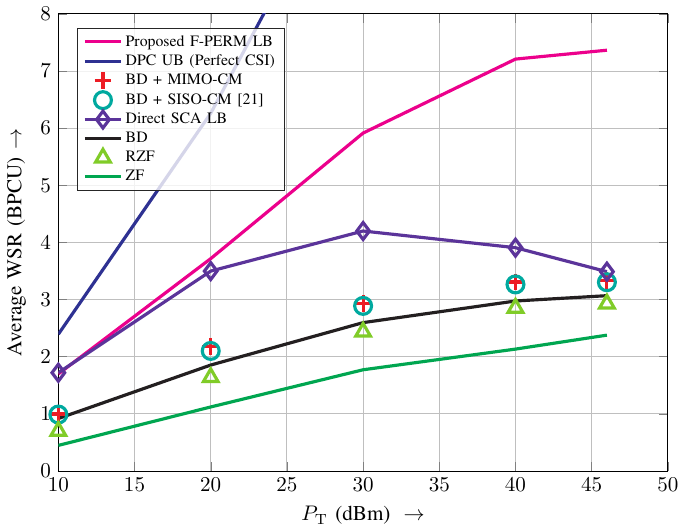}
			\caption{Average WSR for imperfect CSI knowledge as a function of $P_\T$ for $K=6,$ $N=14,$ $M_k=(1,1,2,2,\\4,4),$ $d_k = (250,250,150,150,50,50)\text{ m,}$ $\eta_k \\= (0.3, 0.3, 0.15, 0.15, 0.05, 0.05),$ and $\mu_k= 0.01,\forall\,k,$ and the i.i.d. Gaussian MIMO channel model.}
			\label{fig:wsr-1-1-2-2-4-4-250etc-eta-sik}
		\end{minipage}
	\end{figure*}

	Next, in Figures \ref{fig:wsr-1-1-2-2-4-4-250etc-eta} and \ref{fig:wsr-1-1-2-2-4-4-250etc-eta-sik}, we consider a critically loaded system with $K=6,$ $N=14,$ and $M_k=(1,1,2,2,4,4),$ for perfect and imperfect CSI knowledge at the BS, respectively, and users with unequal distances and weights. The user distances are set to $d_k = (250,250,150,150,50,50)\text{ m,}$ the user weights to $\eta_k = (0.3, 0.3, 0.15, 0.15, 0.05, 0.05),$ and for imperfect CSI, $\mu_k = 0.01,\forall\,k.$ Furthermore, in order to avoid optimization over $6!$ user index combinations, only the fixed-permutation LBs for the WSRs of the proposed SNS-based MIMO-RSMA scheme are presented. From Figures \ref{fig:wsr-1-1-2-2-4-4-250etc-eta} and \ref{fig:wsr-1-1-2-2-4-4-250etc-eta-sik}, we observe that, in this case also, the proposed F-PERM LB significantly outperforms the baseline schemes for both perfect and imperfect CSI. Furthermore, for imperfect CSI and high SNRs, unlike SNS precoding, the performance of Direct SCA LB degrades to that of BD+MIMO-CM because the large transmit power increases the IUI.

	\begin{figure*}
		\centering
		\begin{minipage}[t]{0.48\textwidth}
			\centering
			\includegraphics[width=0.8\textwidth]{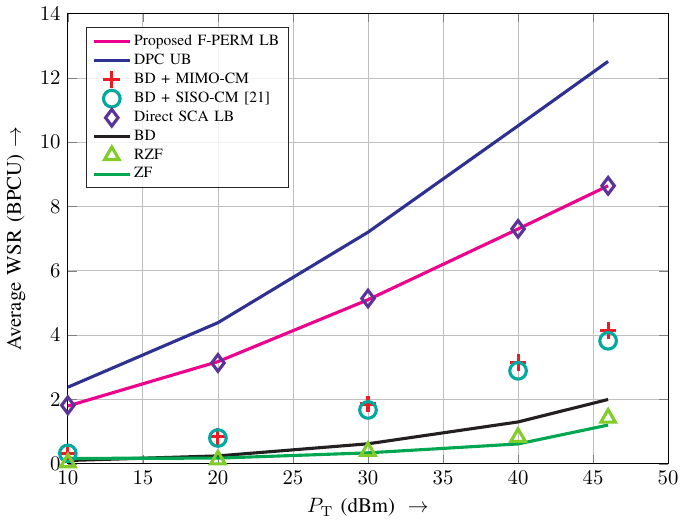}
			\caption{Average WSR for perfect CSI knowledge as a function of $P_\T$ for $K=6,$ $N=14,$ $M_k=(1,1,2,2,\\4,4),$ $d_k = (250,250,150,150,50,50)\text{ m,}$ $\eta_k = (0.3, 0.3, 0.15, 0.15, 0.05, 0.05),$ and the QuaDRiGa channel model \cite{Jaeckel2014}.}
			\label{fig:wsr-1-1-2-2-4-4-250etc-eta-th}
		\end{minipage}%
		\hspace{0.02\textwidth}%
		\begin{minipage}[t]{0.48\textwidth}
			\centering
			\includegraphics[width=0.8\textwidth]{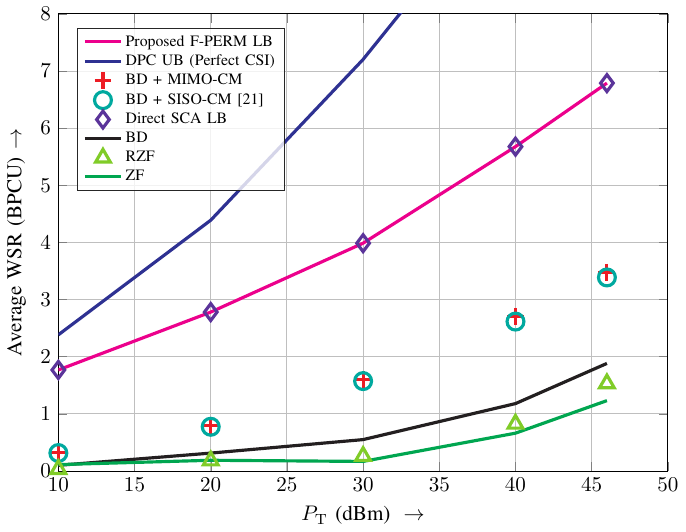}
			\caption{Average WSR for imperfect CSI knowledge as a function of $P_\T$ for $K=6,$ $N=14,$ $M_k=(1,1,2,2,4,4),$ $d_k = (250,250,150,150,50,\\50)\text{ m,}$ $\eta_k = (0.3, 0.3, 0.15, 0.15, 0.05, 0.05),$ $\mu_k= 0.01,\forall\,k,$ and the QuaDRiGa channel model \cite{Jaeckel2014}.}
			\label{fig:wsr-1-1-2-2-4-4-250etc-eta-sik-th}
		\end{minipage}
	\end{figure*}
	
	Next, in Figures \ref{fig:wsr-1-1-2-2-4-4-250etc-eta-th} and \ref{fig:wsr-1-1-2-2-4-4-250etc-eta-sik-th}, we adopt the same system parameters as in Figures \ref{fig:wsr-1-1-2-2-4-4-250etc-eta} and \ref{fig:wsr-1-1-2-2-4-4-250etc-eta-sik} but replace the i.i.d. Gaussian MIMO channel model with the QuaDRiGa channel model \cite{Jaeckel2014}. For imperfect CSI, the elements of $\Delta\bo{H}_k \sim \mathbb{C}^{M_k\times N}, k=1,\dots,K,$ are modeled in the same manner as before. For the BS, we utilize three antenna arrays with $N$ elements, half-wavelength spacing, and $120^\circ$ beamwidth to illuminate the entire $360^\circ$ space. The single-antenna users employ omnidirectional receive antennas, and the multi-antenna users employ uniform linear arrays. The noise variance is set to $\sigma^2 = -90 \text{ dBm.}$ The \texttt{3GPP\_38.901\_UMa} scenario \cite{Jaeckel2014} is utilized for simulation, and the users are distributed randomly and uniformly around the BS. The resulting MIMO channels of the users are \emph{highly correlated.} From Figures \ref{fig:wsr-1-1-2-2-4-4-250etc-eta-th} and \ref{fig:wsr-1-1-2-2-4-4-250etc-eta-sik-th}, we observe that for both perfect and imperfect CSI, owing to the highly correlated channels, RZF, ZF, and BD precoding have vastly inferior performance compared to DPC UB. Furthermore, the proposed SNS-based MIMO-RSMA and Direct SCA LB have similar performances for both perfect and imperfect CSI and significantly outperform the considered baseline schemes. The similar performances for imperfect CSI are due to the fact that, unlike for the i.i.d. Gaussian channels in Figure \ref{fig:wsr-1-1-2-2-4-4-250etc-eta-sik}, for the considered highly correlated channels, not all of the $M_k$ elements of user $k$'s symbol vector $\bo{s}_k,k=1,\dots,K,$ are actually utilized resulting in unused DoFs at the BS, which are readily exploited by both schemes to enhance the robustness against imperfect CSI.
	
	\begin{remark}
		From Figures \ref{fig:wsr-2-4-4-250etc}-\ref{fig:wsr-1-1-2-2-4-4-250etc-eta-sik-th}, we observe that the proposed SNS-based MIMO-RSMA provides a trade-off between performance and computational complexity compared to the considered linear precoding schemes.
	\end{remark}

	\begin{figure*}
		\centering
		\begin{minipage}[t]{0.48\textwidth}
			\centering
			\includegraphics[width=0.8\textwidth]{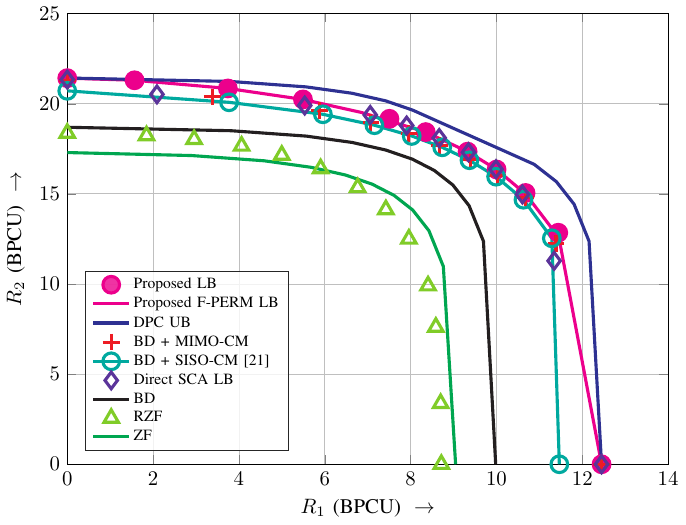}
			\caption{Averaged rate region with perfect CSI knowledge for $K=2,$ $N=4,$ $M_k=2,\forall\,k,$ $d_k = (250,50)\text{ m,}$ and $\eta_1 \in [0,1], \eta_2 = 1-\eta_1.$}
			\label{fig:rr-2-2-250etc}
		\end{minipage}%
		\hspace{0.02\textwidth}%
		\begin{minipage}[t]{0.48\textwidth}
			\centering
			\includegraphics[width=0.8\textwidth]{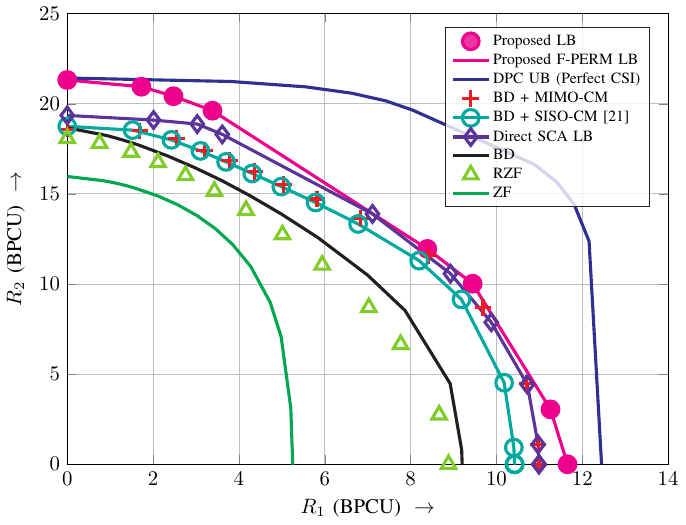}
			\caption{Averaged rate region with perfect CSI knowledge for $K=2,$ $N=4,$ $M_k=2,\forall\,k,$ $d_k = (250,50)\\\text{m,}$ $\eta_1 \in [0,1], \eta_2 = 1-\eta_1,$ and $\mu_k= (0.5,0.001).$}
			\label{fig:rr-2-2-250etc-sik}
		\end{minipage}
	\end{figure*}
	
	Now, for the following figures, we again consider the i.i.d. Gaussian channel model described earlier. In Figures \ref{fig:rr-2-2-250etc} and \ref{fig:rr-2-2-250etc-sik}, we compare the two-user rate regions of the proposed SNS-based MIMO-RSMA scheme with those of the baseline schemes for perfect and imperfect CSI, respectively. We set $d_k = (250,50)\text{ m}$ and $\mu_k= (0.5,0.001).$ For perfect CSI, the proposed LB, BD-based MIMO-RSMA, and Direct SCA LB achieve a significantly larger rate region compared to RZF, ZF, and BD precoding. Nevertheless, for imperfect CSI, the proposed SNS-based MIMO-RSMA scheme outperforms all baseline schemes owing to its enhanced robustness, see Section \ref{sec:iuiimp} for details. Furthermore, for perfect CSI and user rate pairs close to the single user rates of the (farther) user experiencing lower SNR, ZF outperforms RZF, see also \cite{Sung2009}. 
		
	\section{Conclusion}
	\label{sec:con}
	In this paper, we considered the precoder design for an underloaded or critically loaded downlink MU-MIMO communication system employing RS at the transmitter and single-stage SIC at the receivers. We proposed SNS precoding and decoding schemes which utilize linear combinations of the null-space basis vectors of the successively augmented MIMO channel matrices of the users as precoding vectors to adjust the IUI. For perfect CSI knowledge at the BS, we formulated the WSR maximization problem and obtained a feasible LB for the maximum WSR of the proposed SNS scheme via SCA. For imperfect CSI knowledge at the BS, we utilized derivative-based SA to study the robustness of SNS precoding. Our simulation results revealed that the proposed SNS-based MIMO-RSMA scheme outperformed RZF, ZF, and BD precoding as well as BD-based MIMO-RSMA, especially for imperfect CSI, with a $10 \%$ to $50 \%$ performance improvement in terms of the WSR.
	
	Motivated by our results, we believe the following avenues for future work are promising. Firstly, as numerical optimization of the combining matrices $\bo{X}_k,k=1,\dots,K,$ is not desirable in practice, finding closed-form solutions for the combining matrices is crucial for facilitating practical deployment of SNS-based MIMO-RSMA. Furthermore, low-complexity element-by-element decoding schemes are also of interest. Next, in this paper, we assume perfect knowledge of the path loss at the BS and the CSI at the users in order to focus on the impact of fast-varying small-scale fading on SNS precoding itself. In practice, also the path loss at the BS and the CSI at the users are susceptible to estimation errors, which causes the signal reconstruction and SIC to become imperfect. Further analysis of the impact of these imperfections on the performance of SNS-based MIMO-RSMA is of high practical significance. Furthermore, a study of the achievable rate of SNS-based MIMO-RSMA for large antenna arrays as well as  favorable and unfavorable MIMO channel conditions is also of interest. Moreover, the application of the robust power allocation scheme proposed in \cite{Joudeh2016} to the proposed SNS-based MIMO-RSMA is worth investigating. Lastly, other low-complexity user index permutation selection schemes utilizing heuristic methods, as in, e.g., \cite{Tejera2006,Guthy2009,Utschick2018}, are an interesting topic for future studies.
	\begin{appendices}
	\renewcommand{\thesection}{\Alph{section}}
	\renewcommand{\thesubsection}{\thesection.\arabic{subsection}}
	\renewcommand{\thesectiondis}[2]{\Alph{section}:}
	\renewcommand{\thesubsectiondis}{\thesection.\arabic{subsection}:}

	\section{Proof of Proposition \ref{prop:bdcomp}}
	\label{app:bdcomp}
	For user $k,$ the SNS basis vectors for $\bo{\Psi}_k$ can be chosen as $\bo{\Psi}_k = \begin{bmatrix} \bo{\Psi}_k^{\textrm{BD}} & \bo{\Psi}_k^{\textrm{C}}\end{bmatrix},$	where the first $M_k$ columns contain the BD vectors, and the remaining $N_k-M_k$ columns are found with the Gram-Schmidt procedure \cite{Golub2012}. Based on the above definition of $\bo{\Psi}_k,$ $R_{k}$ in (\ref{eqn:rk2}) can be rewritten as in (\ref{eqn:rk2bd}), shown on top of the next page,
	\begin{figure*}
	\begin{align}
		R_{k} &= \log_2\det\Bigg(\bo{I}_{M_k} + \bo{H}_k \Big( 
		\bo{\Psi}_k^{\textrm{BD}} \bo{X}_{k}^{(1)} (\bo{\Psi}_k^{\textrm{BD}})^\H
		+ \bo{\Psi}_k^{\textrm{BD}} \bo{X}_{k}^{(2)} (\bo{\Psi}_k^{\textrm{C}})^\H
		+ \bo{\Psi}_k^{\textrm{C}} \bo{X}_{k}^{(3)} (\bo{\Psi}_k^{\textrm{BD}})^\H \nonumber\\
		&\qquad  + \bo{\Psi}_k^{\textrm{C}} \bo{X}_{k}^{(4)} (\bo{\Psi}_k^{\textrm{C}})^\H \Big) \bo{H}_k^\H  
		\Big[\sigma^2\bo{I}_{M_k} + \bo{H}_k \Big(\sum_{k'=1}^{k-1} \bo{\Psi}_{k'}^{\textrm{C}} \bo{X}_{k'}^{(4)} (\bo{\Psi}_{k'}^{\textrm{C}})^\H\Big)\bo{H}_{k'}^\H\Big]^{-1}\Bigg) \label{eqn:rk2bd}
	\end{align}
	\end{figure*}
	where $\bo{X}_k$ is partitioned as follows:
	\begin{align}
		\bo{X}_k = \; \begin{bNiceMatrix}[first-row,first-col]
			&  M_k             & N_k-M_k \\
			M_k & \bo{X}_{k}^{(1)} & \bo{X}_{k}^{(3)} \\
			N_k-M_k & \bo{X}_{k}^{(2)} & \bo{X}_{k}^{(4)}
		\end{bNiceMatrix}.
	\end{align}
	
	By setting $\bo{X}_{k}^{(2)}, \bo{X}_{k}^{(3)},$ and $\bo{X}_{k}^{(4)}$ in (\ref{eqn:rk2bd}) to zero and restricting the structure of $\bo{X}_{k}^{(1)}$ to $\bo{V}_k^{\textrm{BD}}\bo{D}_k^{\textrm{BD}}(\bo{V}_k^{\textrm{BD}})^\H,$ (\ref{opt:wsr}) reduces to the WSR optimization problem for BD-based MIMO-RSMA. However, since $\bo{X}_{k}^{(1)},$ $\bo{X}_{k}^{(2)},$ $\bo{X}_{k}^{(3)},$ and $\bo{X}_{k}^{(4)}$ can be optimized for SNS-based MIMO-RSMA, we necessarily have $R_\mathrm{wsr}^\star \geq R_\mathrm{wsr}^{\textrm{BD+MIMO-CM}\star}.$ Furthermore, for BD precoding, since $\bo{P}_\ur{c}$ is additionally restricted to $\bo{0},$ we obtain $R_\mathrm{wsr}^\star \geq R_\mathrm{wsr}^{\textrm{BD+MIMO-CM}\star} \geq R_\mathrm{wsr}^{\textrm{BD}\star}.$
	
	Next, for ZF precoding, $\bo{P}_k,k=1,\dots,K,$ in (\ref{eqn:pkzf}) can be simplified using the Schur complement \cite{Horn2012} to $\bo{P}_k = \bo{\Psi}_k^{\textrm{BD}} (\bo{H}_k \bo{\Psi}_k^{\textrm{BD}})^+ (\bo{D}_k^{\textrm{ZF}})^\frac{1}{2}.$ Hence, again, based on (\ref{eqn:rk2bd}) and appropriate substitutions, (\ref{opt:wsr}) can be reduced to the WSR optimization problem for ZF precoding. Therefore, we have $R_\mathrm{wsr}^\star \geq R_\mathrm{wsr}^{\textrm{ZF}\star}.$ \qed

	\section{Proof of Proposition \ref{prop:fo}}
	\label{app:fo}
	In order to prove Proposition \ref{prop:fo}, we first derive the following general result.
	
	\begin{lemma}
		\label{lem:fo}
		Let $m, n > 0, m \leq n,$ and $\bo{A} \in \mathbb{C}^{m\times n}$ be a matrix with full row rank. Furthermore, let $\bo{X},\bo{E} \in \mathbb{C}^{n\times n}$ be positive semi-definite matrices such that $\bo{Y} = \bo{X} + \bo{E}$ is positive semi-definite and in the neighborhood of $\bo{X}.$ Then, the first-order approximation of
		\begin{align}
			f(\bo{Y}) = \blogdet{\bo{I}_m + \bo{A}\bo{Y}\bo{A}^\H} \label{eqn:blogdet}
		\end{align}
		is given by
		\begin{align}
			\tilde{f}(\bo{Y}) &= \blogdet{\bo{I}_m + \bo{A}\bo{X}\bo{A}^\H} \nonumber\\&\quad+ \btr{\bo{A}^\H \Big(\bo{I}_m + \bo{A}\bo{X}\bo{A}^\H\Big)^{-1} \bo{A} (\bo{Y}-\bo{X})}. \label{eqn:tildef}
		\end{align}
	\end{lemma}
	\begin{proof}
		We equip the space of $m \times n$ complex-valued matrices with the inner product: $\langle \bo{A}, \bo{B} \rangle = \btr{\bo{A}^\H \bo{B}}\negmedspace.$ As a first step, we consider the case where $\bo{X}$ and $\bo{E}$ are not necessarily positive semi-definite and (Hermitian) symmetric\footnote{We note that complex-valued positive semi-definite matrices are necessarily Hermitian symmetric.}, and denote the corresponding function as $f_\ur{a}(\bo{Y}).$ By substituting the value of $\bo{Y}$ and simplifying (\ref{eqn:blogdet}), we obtain
		\begin{align}
			f_\ur{a}(\bo{Y}) &= \blogdet{\bo{I}_m + \bo{A}\bo{X}\bo{A}^\H} \nonumber\\&\quad+ \underbrace{\blogdet{\bo{I}_m + \Big(\bo{I}_m + \bo{A}\bo{X}\bo{A}^\H\Big)^{-1}\bo{A}\bo{E}\bo{A}^\H}}_{f_\ur{e}(\bo{E})}.
		\end{align}
		Now, for small $\bnorm{\bo{E}} \ll 1,$ $f_\ur{e}(\bo{E})$ can be expanded in terms of the eigenvalues $\lambda_1,\dots,\lambda_m$ of $\Big(\bo{I}_m +\bo{A}\bo{X}\bo{A}^\H\Big)^{-1}\bo{A}\bo{E}\bo{A}^\H$ as follows:
		\begin{align}
			f_\ur{e}(\bo{E}) &= \sum_{i=1}^{m} \log_\ur{e}(1+\lambda_i) = \sum_{k=1}^{+\infty}\sum_{i=1}^{m} \frac{(-1)^{k+1}}{k}\lambda_i^k \nonumber\\&\overset{(a)}{=} \btr{\bo{A}^\H\Big(\bo{I}_m +\bo{A}\bo{X}\Big)^{-1}\bo{A}\bo{E}} + \bO{\bnorm{\bo{E}}^2}, \label{eqn:frecstep1}
		\end{align}
	where (a) holds because for generic matrices $\bo{X}$ and $\bo{Y},$ $\btr{\bo{X}\bo{Y}} = \btr{\bo{Y}\bo{X}},$ and the higher powers of eigenvalues necessarily depend on the corresponding powers of $\bo{E}.$  From (\ref{eqn:frecstep1}), we conclude that the Fr\'{e}chet derivative \cite{Coleman2012} of $f_\ur{e}(\bo{E})$ is given by $\mathcal{D}_{f_\ur{e}} = \bo{A}^\H\Big(\bo{I}_m + \bo{A}\bo{X}\bo{A}^\H\Big)^{-1}\bo{A}.$ Next, based on \cite[Cor. 2.6]{Srinivasan2020}, for Hermitian symmetric positive semi-definite matrices $\bo{X},$ the Fr\'{e}chet derivative of the corresponding function $f_2(\bo{Y})$ is given by  $\mathcal{D}_{f_2} = \frac{1}{2}\Big(\mathcal{D}_{f_\ur{e}} + (\mathcal{D}_{f_\ur{e}})^\H\Big) \overset{(a)}{=} \mathcal{D}_{f_\ur{e}},$ where (a) holds because $\mathcal{D}_{f_\ur{e}}$ is Hermitian symmetric. Hence, the first-order approximation for $f(\bo{Y})$ is given by
	\begin{align}
		&\blogdet{\bo{I}_m + \bo{A}\bo{X}\bo{A}^\H} + \langle \mathcal{D}_{f_2}, \bo{E} \rangle =\nonumber\\&\qquad \blogdet{\bo{I}_m + \bo{A}\bo{X}\bo{A}^\H} + \btr{\mathcal{D}_{f_2}\bo{E}},
	\end{align}
	from which (\ref{eqn:tildef}) follows.
	\end{proof}
	Lastly, (\ref{eqn:tr12}) and (\ref{eqn:trc}) follow directly by applying Lemma \ref{lem:fo} to (\ref{eqn:rk2}) and (\ref{eqn:rkc2}). \qed

	\section{Proof of Proposition \ref{prop:ikjk}}
	\label{app:ikjk}
	In order to simplify $\|\bo{\Xi}^\uparrow_k\|$ and $\|\bo{\Xi}^\downarrow_k\|, k=1,\dots,K,$ we first propose the following lemma.

	\begin{lemma}
		\label{lemma:ikjk}
		Let $\bo{A}, \Delta\bo{A} \in \mathbb{C}^{m\times n}, m \leq n, m, n > 0,$ denote matrices with full row rank. Furthermore, let $\bo{\Psi}, \bar{\bo{\Psi}} \in \mathbb{C}^{n\times (n-m)}$ denote matrices whose columns are unit-length basis vectors of the null spaces of $\bo{A}$ and $\bar{\bo{A}} = \bo{A} + \Delta\bo{A},$ respectively, and let $\Delta\bo{\Psi} = \bar{\bo{\Psi}} - \bo{\Psi}.$ Then,
		\begin{align}
			\|\bo{A}\bar{\bo{\Psi}}\| &\leq \|\Delta\bo{A}\|, \label{eqn:an}\\
			\min \|\Delta\bo{\Psi}\| &\leq \|\bar{\bo{A}}^+\| \|\Delta\bo{A}\| \nonumber\\
			&\quad-  \log\negmed\big(1-\|\bo{\Psi}^\H(\Delta\bo{A})^\H\big(\bar{\bo{A}}\bar{\bo{A}}^\H\big)^{-1}\Delta\bo{A}\bo{\Psi}\|\big). \label{eqn:nn}
		\end{align} 
	\end{lemma}
		
	\begin{proof}
		Combining $\bo{A}\bo{\Psi} = \bo{0}$ and $\bar{\bo{A}} \bar{\bo{\Psi}} = \bo{0},$ we obtain $\bo{A}\bar{\bo{\Psi}} = -\Delta\bo{A} \bar{\bo{\Psi}},$ from which (\ref{eqn:an}) follows by taking the norm of both sides and utilizing the well-known identities (I.1) $\|\bo{M}_1\bo{M}_2\| \leq \|\bo{M}_1\|\|\bo{M}_2\|$ for generic matrices $\bo{M}_1$ and $\bo{M}_2,$ and (I.2) the norm of a unitary matrix is one. Next, from the standard definition of the Hermitian orthogonal projection, see, e.g., \cite{Horn2012}, we have
		\begin{align}
			\bo{I}_n - \bar{\bo{\Psi}}\bar{\bo{\Psi}}^\H = \bar{\bo{A}}^+ \bar{\bo{A}}, \label{eqn:opba}
		\end{align}
		where $\bar{\bo{A}}^+ \coloneqq \bar{\bo{A}}^\H \big(\bar{\bo{A}}\bar{\bo{A}}^\H\big)^{-1}.$ Right multiplying (\ref{eqn:opba}) by $\bo{\Psi}$ and utilizing again that $\bo{A}\bo{\Psi} = \bo{0},$ we obtain upon simplification
		\begin{align}
			\Delta\bo{\Psi} = -\bar{\bo{A}}^+ \Delta\bo{A}\bo{\Psi} - \bar{\bo{\Psi}} (\Delta\bo{\Psi})^\H \bo{\Psi}, \label{eqn:dn}
		\end{align}
		see also \cite{Edelman1998,Papadopoulo2000} for other equivalent formulations. On the other hand, multiplying (\ref{eqn:opba}) on the left and right by $\bo{\Psi}^\H$ and $\bo{\Psi},$ respectively, and simplifying, we obtain
		\begin{align}
			&(\Delta\bo{\Psi})^\H \bo{\Psi} + \bo{\Psi}^\H \Delta\bo{\Psi} + \bo{\Psi}^\H \Delta\bo{\Psi} (\Delta\bo{\Psi})^\H \bo{\Psi} =\nonumber\\&\qquad - \underbrace{\bo{\Psi}^\H (\Delta\bo{A})^\H \big(\bar{\bo{A}}\bar{\bo{A}}^\H\big)^{-1} \Delta\bo{A} \bo{\Psi}}_{\coloneqq\,\bo{C}}. \label{eqn:long2ord}
		\end{align}
		From (\ref{eqn:long2ord}), a generic expression for $(\Delta\bo{\Psi})^\H \bo{\Psi}$ can be obtained as:
		\begin{align}
			(\Delta\bo{\Psi})^\H \bo{\Psi} = - \bo{Z}\big(\bo{I}_{n-m} -  \bo{C}\big)^{\frac{1}{2}} - \bo{I}_{n-m}, \label{eqn:sol1}
		\end{align}
		where $\bo{Z} \in \mathbb{C}^{(n-m) \times (n-m)}$ is an arbitrary matrix\footnote{Note that matrix $\bo{Z}$ does not need to be unitary.} such that $\bo{Z}^\H \bo{Z} = \bo{I}_{n-m}.$ This degree of freedom is because matrices $\bo{\Psi}$ and $\bar{\bo{\Psi}}$ are not unique, and $\bo{Z}$ can be carefully chosen to minimize $\|(\Delta\bo{\Psi})^\H \bo{\Psi}\|.$ However, for our purposes, it is sufficient to choose $\bo{Z} = -\bo{I}_{n-m}.$ Next, for small $\Delta\bo{A},$ we have \cite{Turnbull1930}:
		\begin{align}
			\big(\bo{I}_{n-m} -  \bo{C}\big)^{\frac{1}{2}} = \bo{I}_{n-m} + \sum_{k = 1}^{+\infty} (-1)^k \prod_{l=0}^{k-1} \left(\frac{1}{2} - l\right) \bo{C}^k.
		\end{align}
		Hence, based on (\ref{eqn:sol1}) with $\bo{Z} = -\bo{I}_{n-m},$ we obtain
		\begin{align}
			\|(\Delta\bo{\Psi})^\H \bo{\Psi}\| &{}\leq \left\|\sum_{k = 1}^{+\infty} (-1)^{k} \prod_{l=0}^{k-1} \left(\frac{1}{2} - l\right) \bo{C}^k\right\| \leq \sum_{k = 1}^{+\infty} \frac{1}{k}\left\|\bo{C}\right\|^k \nonumber\\&{}= -\log(1-\|\bo{C}\|). \label{eqn:dnnbound}
		\end{align}
		Lastly, taking the norm of both sides of (\ref{eqn:dn}), utilizing (I.1) and (I.2) along with the well-known identity (I.3) $\|\bo{M}_1 + \bo{M}_2\| \leq \|\bo{M}_1\| + \|\bo{M}_2\|$ for generic matrices $\bo{M}_1$ and $\bo{M}_2,$ and substituting the above result, we obtain (\ref{eqn:nn}). The $\min$ on the left hand side of (\ref{eqn:nn}) is because the bound holds only for appropriately chosen $\bo{Z};$ e.g., for $\bo{Z} = \bo{I}_{n-m},$ $\|(\Delta\bo{\Psi})^\H \bo{\Psi}\| \approx 2.$
	\end{proof}
	
	The results in (\ref{eqn:normik}) and (\ref{eqn:normjk}) follow by applying Lemma \ref{lemma:ikjk} to (\ref{eqn:i_k}) and (\ref{eqn:j_k}), and utilizing identities (I.1), (I.2), and (I.3). \qed

	\end{appendices}
	
	\makeatletter
	\bibliographystyle{IEEEtran}
	\bibliography{IEEEabrv,references}

\begin{thebibliography}{10}
\providecommand{\url}[1]{#1}
\csname url@samestyle\endcsname
\providecommand{\newblock}{\relax}
\providecommand{\bibinfo}[2]{#2}
\providecommand{\BIBentrySTDinterwordspacing}{\spaceskip=0pt\relax}
\providecommand{\BIBentryALTinterwordstretchfactor}{4}
\providecommand{\BIBentryALTinterwordspacing}{\spaceskip=\fontdimen2\font plus
\BIBentryALTinterwordstretchfactor\fontdimen3\font minus
  \fontdimen4\font\relax}
\providecommand{\BIBforeignlanguage}[2]{{%
\expandafter\ifx\csname l@#1\endcsname\relax
\typeout{** WARNING: IEEEtran.bst: No hyphenation pattern has been}%
\typeout{** loaded for the language `#1'. Using the pattern for}%
\typeout{** the default language instead.}%
\else
\language=\csname l@#1\endcsname
\fi
#2}}
\providecommand{\BIBdecl}{\relax}
\BIBdecl

\bibitem{Krishnamoorthy2021a}
A.~Krishnamoorthy and R.~Schober, ``Successive null-space precoder design for
  downlink {MU-MIMO} with rate splitting and single-stage {SIC},'' \emph{Proc.
  IEEE Intl. Commun. Conf. (ICC) Wrkshp.}, Jun. 2021.

\bibitem{Saad2020}
W.~{Saad}, M.~{Bennis}, and M.~{Chen}, ``A vision of {6G} wireless systems:
  {A}pplications, trends, technologies, and open research problems,''
  \emph{IEEE Network}, vol.~34, no.~3, pp. 134--142, May 2020.

\bibitem{Ding2017}
Z.~{Ding}, X.~{Lei}, G.~K. {Karagiannidis}, R.~{Schober}, J.~{Yuan}, and V.~K.
  {Bhargava}, ``A survey on non-orthogonal multiple access for {5G} networks:
  Research challenges and future trends,'' \emph{IEEE J. Sel. Areas Commun.},
  vol.~35, no.~10, pp. 2181--2195, Oct. 2017.

\bibitem{Krishnamoorthy2020}
A.~Krishnamoorthy, Z.~Ding, and R.~Schober, ``Precoder design and statistical
  power allocation for {MIMO-NOMA} via user-assisted simultaneous
  diagonalization,'' \emph{{IEEE} Trans. Commun.}, vol.~69, no.~2, pp.
  929--945, Nov. 2020.

\bibitem{Makki2020}
B.~{Makki}, K.~{Chitti}, A.~{Behravan}, and M.~{Alouini}, ``A survey of {NOMA}:
  {C}urrent status and open research challenges,'' \emph{IEEE Open J. Commun.
  Soc.}, vol.~1, pp. 179--189, Jan. 2020.

\bibitem{Krishnamoorthy2021}
A.~{Krishnamoorthy} and R.~{Schober}, ``Uplink and downlink {MIMO-NOMA} with
  simultaneous triangularization,'' \emph{{IEEE} Trans. Wireless Commun.},
  vol.~20, no.~6, pp. 3381--3396, Jan. 2021.

\bibitem{Liu2022}
Y.~Liu, S.~Zhang, X.~Mu, Z.~Ding, R.~Schober, N.~Al-Dhahir, E.~Hossain, and
  X.~Shen, ``Evolution of {NOMA} toward next generation multiple access
  {(NGMA)} for {6G},'' \emph{{IEEE} J. Select. Areas Commun.}, vol.~40, pp.
  1037--1071, Apr. 2022.

\bibitem{Mao2018}
Y.~Mao, B.~Clerckx, and V.~O. Li, ``Rate-splitting multiple access for downlink
  communication systems: {B}ridging, generalizing, and outperforming {SDMA} and
  {NOMA},'' \emph{EURASIP J. Wireless Commun. and Netw.}, vol. 2018, no.~1, p.
  133, May 2018.

\bibitem{Zhou2020}
G.~Zhou, Y.~Mao, and B.~Clerckx, ``Rate-splitting multiple access for
  multi-antenna downlink communication systems: {S}pectral and energy
  efficiency tradeoff,'' \emph{IEEE Trans. Wireless Commun. (Early Access)},
  Dec. 2021.

\bibitem{Dizdar2020a}
O.~Dizdar, Y.~Mao, W.~Han, and B.~Clerckx, ``Rate-splitting multiple access for
  downlink multi-antenna communications: {P}hysical layer design and link-level
  simulations,'' \emph{Proc. Intl. Symp. on Personal, Indoor and Mobile Radio
  Commun.}, pp. 1--6, Oct. 2020.

\bibitem{Dizdar2021}
------, ``Rate-splitting multiple access: {A} new frontier for the {PHY} layer
  of {6G},'' \emph{IEEE 92nd Veh. Techn. Conf. (VTC2020-Fall)}, pp. 1--7, Feb.
  2021.

\bibitem{Vishwanath2003}
S.~{Vishwanath}, N.~{Jindal}, and A.~{Goldsmith}, ``Duality, achievable rates,
  and sum-rate capacity of {Gaussian} {MIMO} broadcast channels,'' \emph{IEEE
  Trans. Inf. Theory}, vol.~49, no.~10, pp. 2658--2668, Oct. 2003.

\bibitem{Piovano2016}
E.~Piovano, H.~Joudeh, and B.~Clerckx, ``Overloaded multiuser {MISO}
  transmission with imperfect {CSIT},'' in \emph{Asilomar Conf. on Signals,
  Systems and Computers}, Nov. 2016, pp. 34--38.

\bibitem{Joudeh2017}
H.~Joudeh and B.~Clerckx, ``Rate-splitting for max-min fair multigroup
  multicast beamforming in overloaded systems,'' \emph{{IEEE} Trans. Wireless
  Commun.}, vol.~16, no.~11, pp. 7276--7289, Aug. 2017.

\bibitem{Kaulich2021}
C.~Kaulich, M.~Joham, and W.~Utschick, ``Rate-splitting for the weighted sum
  rate maximization under minimum rate constraints in the {MIMO BC},'' in
  \emph{Proc. IEEE Intl. Commun. Conf. (ICC) Wrkshp.}, Jun. 2021, pp. 1--6.

\bibitem{Joudeh2016}
H.~{Joudeh} and B.~{Clerckx}, ``Sum-rate maximization for linearly precoded
  downlink multiuser {MISO} systems with partial {CSIT}: {A} rate-splitting
  approach,'' \emph{IEEE Trans. Commun.}, vol.~64, no.~11, pp. 4847--4861, Nov.
  2016.

\bibitem{Mao2020}
Y.~{Mao} and B.~{Clerckx}, ``Beyond dirty paper coding for multi-antenna
  broadcast channel with partial {CSIT}: {A} rate-splitting approach,''
  \emph{IEEE Trans. Commun.}, vol.~68, no.~11, pp. 6775--6791, Nov. 2020.

\bibitem{Wiesel2008}
A.~Wiesel, Y.~C. Eldar, and S.~Shamai, ``Zero-forcing precoding and generalized
  inverses,'' \emph{{IEEE} Trans. Signal Processing}, vol.~56, no.~9, pp.
  4409--4418, Aug. 2008.

\bibitem{Peel2005}
C.~Peel, B.~Hochwald, and A.~Swindlehurst, ``A vector-perturbation technique
  for near-capacity multiantenna multiuser communication-part {I:} {Channel}
  inversion and regularization,'' \emph{{IEEE} Trans. Commun.}, vol.~53, no.~1,
  pp. 195--202, Feb. 2005.

\bibitem{Spencer2004}
Q.~H. Spencer, A.~L. Swindlehurst, and M.~Haardt, ``Zero-forcing methods for
  downlink spatial multiplexing in multiuser {MIMO} channels,'' \emph{IEEE
  Trans. Signal Process.}, vol.~52, no.~2, pp. 461--471, Feb. 2004.

\bibitem{Flores2019}
A.~R. {Flores} and R.~C. {de Lamare}, ``Linearly precoded rate-splitting
  techniques with block diagonalization for multiuser {MIMO} systems,''
  \emph{IEEE Intl. Conf. on Commun. (ICC) Wkshp.}, pp. 1--6, May 2019.

\bibitem{Flores2020}
A.~Flores, R.~C. de~Lamare, and B.~Clerckx, ``Study of linear precoding and
  stream combining for rate splitting in {MU-MIMO} systems,'' \emph{arXiv
  preprint, arXiv:2003.03486}, May 2020.

\bibitem{Guthy2009}
C.~Guthy, W.~Utschick, and G.~Dietl, ``Low-complexity linear zero-forcing for
  the {MIMO} broadcast channel,'' vol.~3, no.~6, pp. 1106--1117, 2009.

\bibitem{Utschick2018}
W.~Utschick, C.~Stöckle, M.~Joham, and J.~Luo, ``Hybrid {LISA} precoding for
  multiuser millimeter-wave communications,'' \emph{{IEEE} Trans. Wireless
  Commun.}, vol.~17, no.~2, pp. 752--765, 2018.

\bibitem{Tejera2006}
P.~Tejera, W.~Utschick, G.~Bauch, and J.~Nossek, ``Subchannel allocation in
  multiuser multiple-input–multiple-output systems,'' \emph{{IEEE} Trans.
  Inform. Theory}, vol.~52, no.~10, pp. 4721--4733, 2006.

\bibitem{Kucherenko2016}
S.~Kucherenko and B.~Iooss, ``Derivative-based global sensitivity measures,''
  \emph{Handbook of Uncertainty Quantification}, pp. 1--24, 2016.

\bibitem{Razaviyayn2014}
M.~Razaviyayn, ``Successive convex approximation: {Analysis} and
  applications,'' Ph.D. dissertation, Univ. of Minnesota, May 2014.

\bibitem{Scutari2009}
G.~{Scutari}, D.~P. {Palomar}, and S.~{Barbarossa}, ``The {MIMO} iterative
  waterfilling algorithm,'' \emph{IEEE Tran. Sig. Proc.}, vol.~57, no.~5, pp.
  1917--1935, Jan. 2009.

\bibitem{Wang2007}
C.~Wang, E.~K. Au, R.~D. Murch, W.~H. Mow, R.~S. Cheng, and V.~Lau, ``On the
  performance of the {MIMO} zero-forcing receiver in the presence of channel
  estimation error,'' \emph{{IEEE} Trans. Wireless Commun.}, vol.~6, pp.
  805--810, Mar. 2007.

\bibitem{Eraslan2013}
E.~Eraslan, B.~Daneshrad, and C.-Y. Lou, ``Performance indicator for {MIMO}
  {MMSE} receivers in the presence of channel estimation error,'' \emph{IEEE
  Wireless Commun. Lett.}, vol.~2, pp. 211--214, Apr. 2013.

\bibitem{Sung2009}
H.~Sung, S.-R. Lee, and I.~Lee, ``Generalized channel inversion methods for
  multiuser {MIMO} systems,'' \emph{{IEEE} Trans. Commun.}, vol.~57, no.~11,
  pp. 3489--3499, Nov. 2009.

\bibitem{WangGiannakis2011}
X.~Wang and G.~B. Giannakis, ``Resource allocation for wireless multiuser
  {OFDM} networks,'' \emph{IEEE Trans. Inf. Theory}, vol.~57, no.~7, pp.
  4359--4372, Jun. 2011.

\bibitem{Luo2006}
Z.-Q. Luo and W.~Yu, ``An introduction to convex optimization for
  communications and signal processing,'' \emph{{IEEE} J. Select. Areas
  Commun.}, vol.~24, no.~8, pp. 1426--1438, Jul. 2006.

\bibitem{Boyd2004}
S.~Boyd and L.~Vandenberghe, \emph{Convex {O}ptimization}.\hskip 1em plus 0.5em
  minus 0.4em\relax Cambridge {U}niversity {P}ress, 2004.

\bibitem{Golub2012}
G.~H. Golub and C.~F. Van~Loan, \emph{Matrix {C}omputations}.\hskip 1em plus
  0.5em minus 0.4em\relax JHU Press, 2012, vol.~3.

\bibitem{Wright1997}
S.~J. Wright, \emph{Primal-Dual Interior-Point Methods}.\hskip 1em plus 0.5em
  minus 0.4em\relax Society for Industrial and Applied Mathematics, 1997.

\bibitem{Serrano2015}
S.~A. Serrano, ``Algorithms for unsymmetric cone optimization and an
  implementation for problems with the exponential cone,'' Ph.D. dissertation,
  Stanford University, 2015.

\bibitem{Krishnamoorthy2013}
A.~Krishnamoorthy and D.~Menon, ``Matrix inversion using {Cholesky}
  decomposition,'' in \emph{IEEE Signal Proc.: Algorithms, Architectures,
  Arrangements, and Applications (SPA) Conf.}, Sep. 2013, pp. 70--72.

\bibitem{Raghavan2017}
V.~Raghavan, S.~Subramanian, J.~Cezanne, A.~Sampath, O.~H. Koymen, and J.~Li,
  ``Single-user versus multi-user precoding for millimeter wave {MIMO}
  systems,'' \emph{{IEEE} J. Select. Areas Commun.}, vol.~35, no.~6, pp.
  1387--1401, Mar. 2017.

\bibitem{Jaeckel2014}
S.~Jaeckel, L.~Raschkowski, K.~Börner, and L.~Thiele, ``Quadriga: A {3-D}
  multi-cell channel model with time evolution for enabling virtual field
  trials,'' \emph{{IEEE} Trans. Antennas Propagat.}, vol.~62, no.~6, pp.
  3242--3256, Mar. 2014.

\bibitem{Liu2008}
J.~Liu, Y.~T. Hou, and H.~D. Sherali, ``On the maximum weighted sum-rate of
  {MIMO} {Gaussian} broadcast channels,'' in \emph{IEEE Intl. Conf. on Commun.
  (ICC)}, May 2008, pp. 3664--3668.

\bibitem{Gamal2011}
A.~Gamal and Y.~Kim, \emph{Network {I}nformation {T}heory}.\hskip 1em plus
  0.5em minus 0.4em\relax Cambridge University Press, 2011.

\bibitem{Horn2012}
R.~A. Horn and C.~R. Johnson, \emph{Matrix {A}nalysis}.\hskip 1em plus 0.5em
  minus 0.4em\relax Cambridge university press, 2012.

\bibitem{Coleman2012}
R.~Coleman, \emph{Calculus on Normed Vector Spaces}.\hskip 1em plus 0.5em minus
  0.4em\relax Springer New York, 2012.

\bibitem{Srinivasan2020}
S.~Srinivasan and N.~Panda, ``What is the gradient of a scalar function of a
  symmetric matrix?'' \emph{arXiv preprint arXiv:1911.06491}, May 2020.

\bibitem{Edelman1998}
A.~Edelman, T.~A. Arias, and S.~T. Smith, ``The geometry of algorithms with
  orthogonality constraints,'' \emph{SIAM Journal on Matrix Analysis and
  Applications}, vol.~20, no.~2, pp. 303--353, 1998.

\bibitem{Papadopoulo2000}
T.~Papadopoulo and M.~I.~A. Lourakis, ``Estimating the {Jacobian} of the
  singular value decomposition: {Theory} and applications,'' in \emph{Computer
  Vision - ECCV 2000}.\hskip 1em plus 0.5em minus 0.4em\relax Berlin,
  Heidelberg: Springer Berlin Heidelberg, Apr. 2000, pp. 554--570.

\bibitem{Turnbull1930}
H.~W. Turnbull, ``A matrix form of {T}aylor's theorem,'' \emph{Proceedings of
  the Edinburgh Mathematical Society}, vol.~2, no.~1, p. 33–54, Jan. 1930.

\end{thebibliography}
	
\end{document}